\documentclass{lmcs}
\pdfoutput=1

\usepackage{lastpage}
\lmcsdoi{22}{3}{6}
\lmcsheading{}{\pageref{LastPage}}{}{}%
{Sep.~30,~2025}{Aug.~12,~2026}{}

\usepackage{misc}
\usepackage{someCustomMacra}
\usepackage{xspace}
\usepackage{amsfonts}
\usepackage{amsmath}
\usepackage{amssymb}
\usepackage{amsthm}
\usepackage{microtype}
\usepackage{csquotes}
\usepackage{tikz}
\usepackage{textcomp}
\usepackage{stmaryrd}
\usepackage{wasysym}
\usepackage{soul}
\usepackage[T1]{fontenc}
\usepackage[utf8]{inputenc}
\usepackage{hyperref}
\usepackage{microtype}

\usetikzlibrary{arrows, arrows.meta, automata, matrix, calc, positioning, shapes.geometric, decorations.pathreplacing}

\newcommand{\sig}{\ensuremath{\text{sig}}\xspace}
\newcommand{\ATM}{\ensuremath{\mathbb{A}}\xspace}


\begin{document}

\bibliographystyle{alphaurl}

\title[The Complexity of Defining and Separating Fixpoint Formulae in Modal Logic]{The Complexity of Defining and Separating\texorpdfstring{\\}{} Fixpoint Formulae in Modal Logic}

\author[J.Ch.~Jung]{Jean Christoph Jung\lmcsorcid{0000-0002-4159-2255}}[a]
\author[J.~Kołodziejski]{Jędrzej Kołodziejski\lmcsorcid{0000-0001-5008-9224}}[b]

\address{Department of Computer Science, TU Dortmund University, Germany} 
\email{jean.jung@cs.tu-dortmund.de}  

\address{Faculty of Mathematics, Informatics and Mechanics, University of Warsaw} 
\email{j.kolodziejski@uw.edu.pl}  

%
%
%

%

\begin{abstract}
  Modal separability for modal fixpoint formulae is the problem to decide
  for two given modal fixpoint formulae $\varphi,\varphi'$ whether there is
  a modal formula $\psi$ that separates them, in the sense that
  $\varphi\models\psi$ and $\psi\models\neg\varphi'$. We study modal
  separability and its special case modal definability over various
  classes of models, such as arbitrary models, finite models, trees,
  and models of bounded outdegree. Our main results are that modal
  separability is \PSpace-complete over words, that is, models of outdegree
  $\leq 1$, \ExpTime-complete over unrestricted and over binary
  models, and \TwoExpTime-complete over models of outdegree bounded by
  some $d\geq 3$. Interestingly, this latter case behaves
  fundamentally different from the other cases also in 
  that modal logic does not enjoy the Craig interpolation property
  over this class. Motivated by this we study also the induced
  interpolant existence problem as a special case of modal
  separability, and show that it is \coNExpTime-complete and thus
  harder than validity in the logic. Besides deciding separability, we also provide algorithms for the effective construction of separators. Finally, we consider in a case study the
  extension of modal fixpoint formulae by graded modalities and
  investigate separability by modal formulae and graded modal
  formulae.
\end{abstract}

\maketitle

\section{Introduction}\label{sec:intro}

For given logics $\Lmc,\Lmc^+$, the \emph{$\Lmc$-separability problem
for $\Lmc^+$} is to decide, given two $\Lmc^+$-formulae $\varphi,\varphi'$,
whether there is an \Lmc-formula $\psi$ that separates
$\varphi$ and $\varphi'$ in the sense that $\varphi\models\psi$ and
$\psi\models\neg\phi'$. Obviously, a separator can only exist when
$\varphi$ and $\varphi'$ are inconsistent: if $\varphi\wedge\varphi'$ is satisfiable in some model $\Mmc$, then $\Mmc$ satisfies both any separator $\psi$ (since $\varphi\models \psi$) and its negation $\neg \psi$ (since $\psi\models\neg\varphi'$ and thus $\varphi'\models\neg \psi$). Moreover, when $\Lmc\supseteq\Lmc^+$ and $\varphi,\varphi'$ are inconsistent, we can take $\varphi$ or $\varphi'$ as trivial separators, so the separability problem is nontrivial only when $\Lmc$ cannot express the entire $\Lmc^+$. Finally note that, for logics $\Lmc^+$ closed under negation, \Lmc-separability generalizes the \emph{\Lmc-definability problem for $\Lmc^+$}: decide
whether a given $\Lmc^+$-formula is equivalent to an $\Lmc$-formula.
Indeed, $\varphi\in \Lmc^+$ is equivalent to an \Lmc-formula iff
$\varphi$ and $\neg\varphi$ are \Lmc-separable. Since separability is more
general than definability, solving it requires an even better understanding of the logics under
consideration.

Both definability and separability are central problems in logics, and both have a wide range of applications in computer science. On the one hand, $\Lmc$-definability in $\Lmc^+$ is the essential question in understanding the expressive power of $\Lmc$ within $\Lmc^+$. On the other hand, some applications require understanding the more general notion of separability, as the following examples illustrate:  


\begin{enumerate}

    \item In \emph{verification}, a separator can be viewed as a property that
    separates good behavior, given by $\varphi$, from bad behavior, given by
    $\varphi'$. In this case $\varphi,\varphi'$ are not necessarily logical formulae but could be other symbolic representations  of sets of reachable states. The separator can
    sometimes act as an easy abstraction of the reachable states, for instance, in the
    construction of invariants. This view of verification in terms of separability has induced the study of separability in formal languages. As seminal work in this area let us only mention definability and separability of regular word languages by languages recognized by first-order logic~\cite{DBLP:journals/corr/PlaceZ14,
DBLP:journals/iandc/Schutzenberger65a,DBLP:journals/tcs/ChoH91}.

    \item Separators can be used for \emph{explaining} inconsistencies or, dually, entailment. For the latter, we assume a given entailment $\varphi\models\varphi'$ in a \enquote{complicated} logic $\Lmc^+$. Then $\varphi,\neg\varphi'$ are inconsistent and a separator $\psi$ for $\varphi,\neg\varphi'$ in an \enquote{easier} logic $\Lmc$ can serve as a good explanation for the complicated entailment $\varphi\models \varphi'$ since $\varphi\models\psi$ and $\psi\models\varphi'$. Example~\ref{ex:sep1} below provides a more concrete example of this explanatory aspect. Separability has been recently studied for logics supporting counting or constants, respectively, by logics without this support~\cite{separatingcounting,DBLP:conf/lics/JungKW26}.

    \item \emph{(Craig) interpolation} can be thought of as an instance of separation, when the logic~\Lmc is defined in terms of restricted signature. Craig interpolation is clearly central to computer science, see the upcoming volume~\cite{volumecraig} and references therein. Beyond this rather informal relation between interpolation and separation, there is also a strong technical connection, which we shall exploit later in the paper.

    \item In \emph{learning logical formulae from labeled data examples}, one is given a knowledge base $\Omc$, formulated in language $\Lmc^+$, and sets $P,N$ of positive and negative examples and asks whether there is a formula from a query language $\Lmc$ \emph{fitting} to the positive and negative examples. One possible formalization of fitting is to require from the fitting formula $\psi$ that $\Omc\models \psi(\vec a)$, for every $\vec a\in P$, and $\Omc\models\neg\psi(\vec a)$, for all $\vec a\in N$~\cite{DBLP:journals/ai/JungLPW22}. Hence, $\psi$ plays the role of a separator between positive and negative examples. 
    
\end{enumerate}



In this paper we study definability and separability of formulae of
the modal $\mu$-calculus
$\muML$~\cite{DBLP:conf/focs/Pratt81,DBLP:journals/tcs/Kozen83} by
formulae in propositional modal logic $\ML$. $\muML$ is the extension
of $\ML$ with fixpoints that encompasses virtually all
specification languages such as
$\PDL$~\cite{DBLP:journals/jcss/FischerL79} and $\LTL$ and
$\CTL$~\cite{DBLP:books/daglib/0020348}. Hence, when we ask for separators of given inconsistent $\muML$-formulae $\varphi,\psi$ by an $\ML$-formula, we ask the question whether the recursion in the form of fixpoints in $\varphi,\psi$ is crucial to explain their inconsistency. Let us consider an example.

\begin{exa} \label{ex:sep1} Consider the following 
  properties $P_1,P_2,P_3$ of vertex-labelled trees: 
  \begin{itemize}

    \item $P_1$: there is an infinite path starting in the root on which each point
      satisfies $a$;

    \item $P_2$: on every path there are only finitely many points
      satisfying $a$;

    \item $P_3$: on every path at most two points satisfy $a$.

  \end{itemize}
  The properties are expressible in $\muML$ but not in $\ML$, and both of the pairs
  $P_1,P_2$ and $P_1,P_3$ are mutually exclusive, and hence $P_1$ entails both the complement of $P_2$ and of $P_3$. The properties $P_1,P_3$ are separated by the $\ML$-formula $\psi=a \wedge
  \Diamond (a\wedge \Diamond a)$ which expresses that there is a
  path starting with three points satisfying $a$. As discussed in Point~(2) above, $\psi$ may serve as an intuitive explanation for the entailment of the complement of $P_3$ from $P_1$.
  On the other hand, separators are not guaranteed to exist, and indeed, no $\ML$-formula separates
  $P_1,P_2$. The intuitive reason for this is that any $\ML$-formula
  $\psi$ only sees trees up to depth $|\psi|$, and one can find two
  trees with properties $P_1,P_2$ which nonetheless look the same up to depth $|\psi|$. A formal proof of non-existence of an $\ML$-separator is given in Example~\ref{ex:sep2} in the paper.\qed
%
%
\end{exa}
We explore the definability and separability problems over several classes of models relevant for computer science: all models, words, trees of bounded or unbounded outdegree; as well as restrictions of all these classes to finite models.
On top of analyzing the decision problems, we also address the problem
of constructing definitions and separators whenever they
exist. The starting point for our research is the seminal paper of
Otto~\cite{DBLP:conf/stacs/Otto99}, where he solves modal definability
over models of bounded and unbounded outdegree. In this paper, we
continue this line of work, extend it to separability, and establish a fairly complete and
interesting picture. 
Table~\ref{tab:mainresults} summarizes our
results. We now explain its content further.

\begin{table}
  \centering

  \begin{tabular}{c|c|c|c|c}
     & all models & words & binary trees &
     $d$-ary trees, $d\geq 3$ \\
    \hline
    $\ML$-def.& \ExpTime~\cite{DBLP:conf/stacs/Otto99} & \PSpace & \ExpTime~\cite{DBLP:conf/stacs/Otto99} & \ExpTime~\cite{DBLP:conf/stacs/Otto99} \\
    $\ML$-sep. & \ExpTime & \PSpace & \ExpTime &
    \TwoExpTime \\
    sep. constr. & double exp. & single exp. & double
    exp. & triple exp. \\
    interpolant ex. & always & always & always &
    \coNExpTime\\
  \end{tabular}
  \caption{Overview of our results. All complexity results are
  completeness results.}
  \label{tab:mainresults}
\end{table}
 
The first line essentially repeats Otto's results; we only add the
observation that $\ML$-definability over words is \PSpace-complete. The second line contains our results for $\ML$-separability. The case of words is the easiest one, both in terms of computational complexity
and required arguments.  Next come the cases of binary and of
unrestricted trees. These two classes possess some nice structural
properties which (although true for different reasons) enable a common
algorithmic treatment. Finally, the cases of trees with outdegree
bounded by a number $d\geq3$ enter the stage. These trees lack the
good properties essential for previous constructions, which results in
significantly higher computational complexity. The hardness result for $d\geq 3$ is
interesting for two reasons. First, as it is entirely standard
to encode trees of higher outdegree into binary ones, one could expect
the ternary (and higher) case to have the same complexity as the
binary one. And second, even though there are known cases when separation is provably harder than definability (regularity of visibly pushdown languages is decidable~\cite[Theorem 19]{Loeding19} but regular separability thereof is not~\cite[Theorem 2.4]{DBLP:conf/lics/Kopczynski16}), to the best of our knowledge our results are the only such case known in the area of logic.
 
The complexity landscape for deciding separability is also reflected in
the maximal sizes of the separators that we construct. Relying on the
well-known connection of $\muML$ to automata, we provide
effective constructions for the cases of all models, words, and
binary trees. It is worth mentioning that equally effective constructions for
definability over all models are given in~\cite{DBLP:conf/csl/LehtinenQ15}, but they do not work for
separability. The ternary case follows from a general argument.
Our construction of separators over words is optimal. Under mild assumptions (there are at least two modalities) the constructions over binary and over unrestricted trees are optimal as well, but we leave it open whether these assumptions are needed for the lower bounds. In the case of ternary and higher outdegree trees we only conjecture optimality of the constructed separators.

Finally, we observe that $\ML$ lacks the Craig interpolation property
over trees of outdegree bounded by $d\geq 3$. Recall that a
\emph{Craig interpolant for $\phi\models\phi'$} in some logic $\Lmc$
is a formula $\psi\in \Lmc$ only using the common symbols of $\phi$
and $\phi'$ and such that $\phi\models\psi\models\phi'$. A logic
satisfies the \emph{Craig interpolation property (CIP)} if a Craig
interpolant of $\phi\models\phi'$ always exists. It is known that
$\ML$ enjoys CIP over all models and over
words~\cite{10.1007/BFb0059541} and it follows from our techniques
that this transfers to binary trees. In contrast and as mentioned
above, over ternary and higher-arity trees $\ML$ lacks the CIP. It
is worth mentioning that modal logic over frames of arity bounded by
some $d$ has been studied under the name
$\mathbf{K}\oplus\textit{\textbf{alt}}_d$~\cite{DBLP:journals/aml/Bellissima88}. 
Our results imply that $\mathbf{K}\oplus\textit{\textbf{alt}}_d$ enjoys CIP iff
$d\leq 2$. Motivated by the
lack of CIP over higher-arity trees, we study the induced interpolant
existence problem -- determining whether two given $\ML$-formulae
$\phi,\phi'$ admit a Craig interpolant -- as a special case of
separability. We show it to be \coNExpTime-complete over higher
arity trees, and thus harder than validity. Interpolant existence has
recently been studied for other logics without CIP and our proofs are inspired by the techniques used there~\cite{DBLP:conf/lics/JungW21,DBLP:journals/tocl/ArtaleJMOW23}.

As an application of our results for $d$-ary trees with $d\geq 3$ we
additionally present a case study: definability and separability in the \emph{graded}
setting in which we allow modalities saying ``there are at
least $k$ children such that
[...]''~\cite{DBLP:journals/ndjfl/Fine72}. Such graded modalities are a
standard extension of modal logic that is especially relevant in
applications in knowledge representation for conceptual
modeling~\cite{DBLP:books/daglib/0041477}. Hence, the larger logic $\Lmc^+$ is the \emph{graded $\mu$-calculus} in this case, see for example~\cite{DBLP:conf/cade/KupfermanSV02,DBLP:conf/ijcai/CalvaneseGL99}. For the separating language $\Lmc$, we consider two options. In the first, we allow graded modalities also in the separating language \emph{graded modal logic $\grML$}, whereas in the second we stick to $\ML$ as before. The first setting behaves essentially as the non-graded case in the sense that both $\grML$-definability and $\grML$-separability are \ExpTime-complete. More interesting is the second option, where definability and separability again show different behavior. While
$\ML$-definability of $\mugrML$-formulae is \ExpTime-complete, $\ML$-separability is \TwoExpTime-complete. The intuitive reason for the higher complexity in the latter case is that trees of bounded arity are definable in graded $\muML$. These results are also related with a recent study about separating logics supporting counting quantifiers by logics without these~\cite{separatingcounting}. 

 

It is worth to mention that $\ML$-definability of $\muML$-formulae
generalizes the \emph{boundedness problem} which asks whether a
formula with a single fixpoint is equivalent to a modal formula. 
Boundedness has been studied for other logics such as 
monadic-second order logic~\cite{DBLP:journals/corr/BlumensathOW14}, datalog~\cite{DBLP:journals/jlp/HillebrandKMV95}, and
the guarded fragment of first-order
logic~\cite{DBLP:conf/lics/BenediktCCB15}. Our paper is an extension
of the preliminary papers~\cite{DBLP:conf/dlog/JungK24,DBLP:conf/stacs/JungK25}.

The paper is organized as follows. After this introduction, we set notation and recall basic facts in the preliminary Section~\ref{sec:prelims}. Next, we formally introduce the main notions of definability and separability, discuss a relevant construction of Otto, and solve the case of all models in Section~\ref{sec:foundations}. In the following Sections~\ref{sec:unary} and~\ref{sec:binary} we deal with unary and binary trees, and in Section~\ref{sec:ternary} we solve the most challenging case of trees of outdegree bounded by $d\geq 3$. Section~\ref{sec:graded} applies our results to the case with graded modalities. The last Section~\ref{sec:conclusion} contains conclusions and final remarks.




\section{Preliminaries}\label{sec:prelims}

We recall the main notions about modal logic $\ML$ and the modal
$\mu$-calculus $\muML$. For the rest of this paper fix disjoint,
countably infinite sets $\Propositions$ of \emph{propositions} and $\Var$ of \emph{variables}.  
The syntax of $\muML$ is given by the rule 
\[
  \phi ::= \atProp\ |\ \neg\atProp\ |\ \phi\vee\phi\ |\ \phi\wedge\phi\ |\ \Diamond\phi\ |\ \Box\phi\ |\ x\ |\ \mu x.\phi\ |\ \nu x.\phi
\]
where $\atProp\in\Propositions$ and $x\in \Var$.
We assume that formulae of $\muML$ are in a normal form such that every $x\in\Var$ appears at
most once in a formula, and if it does appear then its appearance has
a unique superformula $\psi$ beginning with $\mu x$ or $\nu x$. Modal
logic $\ML$ is defined as the fragment of $\muML$ with no fixpoint
operators $\mu$ and $\nu$ nor variables. In both $\ML$ and $\muML$, we
use abbreviations like $\top$ (for $a\vee \neg a$ for some
$a\in\Propositions$), $\Diamond^n\varphi$ (for a formula
$\Diamond\ldots\Diamond\varphi$ with $n$ leading $\Diamond$'s), and
$\neg \varphi$. A \emph{signature} is a finite subset $\sigma\subset\Propositions$ of propositions. We denote with $\sig(\varphi)$ the
set of propositions that occur in $\varphi$. As usual, the
\emph{modal depth} of an $\ML$-formula is the maximal nesting depth of
$\Diamond,\Box$. With $\ML^n$ we denote the class of all $\ML$-formulae of modal depth at most $n$, and
with $\ML^n_\sigma$ we denote its subclass restricted to signature $\sigma$. The set of subformulae of $\varphi$ is denoted with $\SubFor(\varphi)$. The \emph{size} $|\phi|$ of a formula $\phi$ is the length of $\phi$ represented as a string. This choice of the simplest possible measure of size does not matter for most of our results.
We will briefly discuss alternative notions of size in the concluding Section~\ref{sec:conclusion}.

Both $\ML$ and $\muML$ are interpreted in pointed 
Kripke structures. More formally, a model $\M$ is a quadruple
$\M=(M,\point_I,{\to},\val)$ consisting 
of a set $M$ called
its \emph{universe}, a distinguished point $\point_I\in M$ called
the \emph{root}, an accessibility relation
${\to}\subset M\times M$, and a valuation
$\val:M\to\powerset{\Propositions}$.

The semantics of $\muML$ can be defined in multiple equivalent ways. The one most convenient for us is through parity games. These are infinite-duration games with perfect information, played between players $\eve$ and $\adam$. Parity games provide a convenient definition of the $\muML$ semantics, and will be useful in some automata-theoretic constructions. We employ a standard terminology regarding parity games, such as positions, finite and infinite plays, strategies etc.~(see~\cite{Venema20} for an introduction with a special focus on $\muML$). Additionally, given a strategy $\strategy$ for either player we call a play $\pi$ a \emph{$\strategy$-play} if $\pi$ is consistent with $\strategy$.
Towards the game-theoretic definition of the semantics of $\muML$ assume a model $\M$ and a formula $\phi\in\muML$. We define a parity game called the \emph{semantic game for $\M$ and $\phi$} and denoted $\game(\M,\phi)$. The positions are $M\times\SubFor(\phi)$. The moves depend on the topmost connective. From a position of the shape $(\point,\psi\vee\psi')$ or $(\point,\psi\wedge\psi')$ it is allowed to move to either $(\point,\psi)$ or $(\point,\psi')$. From $(\point,\Diamond\psi)$ and $(\point,\Box\psi)$ the allowed moves lead to all $(\altpoint,\psi)$ such that $\point\rightarrow\altpoint$.
In position $(\point,\atProp)$ or $(\point,\neg\atProp)$ the game stops and $\eve$ wins iff $\point$ satisfies the formula component $\atProp$ or $\neg\atProp$, respectively. From $(\point,\mu x.\psi)$ and $(\point,\nu x.\psi)$ the game moves to $(\point,\psi)$, and from $(\point,x)$ to $(\point,\theta)$ where $\theta$ is the unique superformula of $x$ beginning with $\mu x$ or $\nu x$. $\eve$ owns positions whose formula component has $\vee$ or $\Diamond$ as the topmost connective and $\adam$ owns all other positions. $\eve$ wins an infinite play $\pi$ if the outermost subformula seen infinitely often in $\pi$ begins with $\nu$. We say that $\M,\point$ \emph{satisfies} $\phi$ and write $\M,\point\models\phi$ if $\eve$ wins the game $\game(\M,\phi)$ from position $(\point,\phi)$. Since $\M$ is by definition pointed, we abbreviate $\M,\point_I\models \varphi$ with $\M\models \varphi$.

%
 
The same symbol denotes entailment:
$\phi\models\psi$ means that every model of $\phi$ is a model of
$\psi$. In the case only models from some fixed class $\ClassMod$ 
are considered we talk about satisfiability and entailment
\emph{over} $\ClassMod$. Let $\LL$ be a subset of $\muML$ such as $\ML$ or $\ML_\sigma^n$. If two models $\M$ and $\N$ satisfy the same formulae of $\LL$ then we call them $\LL$-equivalent and write $\M\equiv_\LL\N$.

In the paper we will study models of bounded and unbounded outdegree.  The \emph{outdegree} of a point $w\in M$ in a model
$\M=(M,\point_I,\to,\val)$ is the number of successors of $w$ in
the underlying directed graph $G_\M = (M,\to)$. We say that $\M$ has
\emph{bounded outdegree} if there is a finite uniform upper bound $d$
on the outdegree of its points. In the latter case, we will call $\M$
\emph{$d$-ary}, and \emph{binary} or \emph{ternary} if $d=2$ or $d=3$.
If $d=1$, then we call $\M$ a \emph{word}.
A $d$-ary model is \emph{full} if each of its nodes is either a leaf
(i.e.~has no children) or has precisely $d$ children.
A model $\M$ is a \emph{tree} if $G_\M$ is a (directed) tree with root
$v_I$. We denote with $\TT^d$ the class of all $d$-ary tree models.
Both $\ML$ and $\muML$ are invariant under bisimulation, and every ($d$-ary) model is bisimilar to a ($d$-ary) tree. Hence, we do not lose generality by only looking at tree models.

A \emph{prefix} of a tree is a subset of its universe closed under
taking ancestors. When no confusion arises we identify a prefix
$N\subset M$ with the induced subtree $\N$ of $\M$ that has $N$ as its
universe. The \emph{depth} of a point is the distance from the root.
The prefix of depth $n$ (or just \emph{$n$-prefix}) is the set of all
points at depth at most $n$ and is denoted by $M_{|_n}$ (and the
corresponding subtree by $\M_{|_n}$).

\subsection{Bisimulations}

We define bisimulations and bisimilarity for trees, assuming for
convenience that bisimulations only
link points at the same depth. This is without loss of generality, since for bisimilar trees there is always a bisimulation that relates only points at the same depth. Let $\M,\M'$ be trees and
$Z\subseteq M\times M'$ a relation between $M$ and $M'$ that relates
only points of the same depth. Then, $Z$ is a \emph{bisimulation
between $\M$ and $\M'$} if it links the roots $\point_I Z\point_I'$, and for every $wZw'$ the following conditions are satisfied:
%
\begin{description}

  \item[$(\baseCond)$] $\val(w) = \val'(w')$, 

  \item[$(\forthCond)$] for every $v\in M$ with $w\to v$ there is a
    $v'\in M'$ with $w'\to v'$ and $vZv'$, and 

  \item[$(\backCond)$] for every $v'\in M'$ with $w'\to v'$ there is a
    $v\in M$ with $w\to v$ and $vZv'$.

\end{description}
A \emph{functional bisimulation} (also known as \emph{bounded morphism}) is a function whose graph is a bisimulation. If $Z$ is a functional bisimulation from $\M$ to $\M'$ then we write $Z:\M\bisFun{}\M'$ and call $\M'$ a \emph{bisimulation quotient} of $\M$. The \emph{bisimilarity quotient} of $\M$ is a quotient $\M'$ of $\M$ such that if $Z':\M'\bisFun\M''$ then $Z'$ is bijective. 
It follows from analogous results for arbitrary models that every tree $\M\in\TT^d$ has a unique (up to isomorphism) bisimilarity quotient $\M'\in\TT^d$ and that two trees are bisimilar iff their bisimilarity quotients are isomorphic.

Further, for every $n\in\NN$ and every subset $\sigma\subset\Propositions$ of the signature we consider a restricted variant of bisimulations called \emph{$(\sigma,n)$-bisimulations}.
In a \emph{$(\sigma,n)$-bisimulation} the $\baseCond$ condition is
only checked with respect to $\sigma$ and the $\backCond$ and
$\forthCond$ conditions only for points at depth smaller than $n$.
Formally, a relation $Z\subset M\times M'$ is a
$(\sigma,n)$-bisimulation if it is a bisimulation between the
$n$-prefixes of the $\sigma$-reducts of $\M,\M'$. We call a
$(\sigma,n)$-bisimulation between $\M,\M'$ a
\emph{$(\sigma,n)$-isomorphism} if it is bijective on the $n$-prefixes
of $\M,\M'$.
We write $\M\bis_\sigma^n\M'$ if there exists a
$(\sigma,n)$-bisimulation between $\M$ and $\M'$ and
$\M\cong_\sigma^n\M'$ if there is a $(\sigma,n)$-isomorphism between
them. Crucially, over every class $\ClassMod$ of models and for every finite
$\sigma$ the equivalences $\equiv_{\ML_\sigma^n}$ and $\bis_\sigma^n$
coincide for every $n$~\cite{DBLP:books/el/07/GorankoO07}.

%

\subsection{Automata} 

We exploit the well-known connection of $\muML$ and automata that read
tree models. A
\emph{nondeterministic parity tree automaton (NPTA) over $d$-ary trees $\TT^d$} is a tuple
$\A=(Q,\Sigma,q_I,\delta,\rank)$ where $Q$ is a finite set of states,
$q_I\in Q$ is the initial state, $\Sigma=\powerset{\sigma}$ for some signature $\sigma\subseteq \Propositions$, $\rank$ assigns each state a
priority from $\mathbb{N}$, and $\delta$ is a transition function of type:
\[
\delta:Q\times\alphabet \to \powerset{Q^{\leq d}},\]
where $Q^{\leq d}$ denotes the set of all tuples over $Q$ of length
at most $d$.
A run of $\A$ on a tree $\M$ is an assignment $\rho:M\to Q$ sending
the root of the tree to $q_I$ and consistent with $\delta$ in the
sense that
$(\rho(\point_1),...,\rho(\point_k))\in\delta(\rho(\point),\val(\point)\cap \sigma)$ for every point $\point$ with children $\point_1,...,\point_k$.
On occasion when considering trees of unbounded outdegree we will use
automata with transition function of type
\[\delta:Q\times\alphabet\to\powerset{\powerset{Q}}.\] Then, consistency
of $\rho$ with $\delta$ means that $\{\rho(\point')\ |\ \point'\in
V\}\in\delta(\rho(\point),\val(\point)\cap\sigma)$ for every $\point$ with a set $V$ of children.
In either case, we call the run $\rho$ \emph{accepting} if for every infinite path $v_0,v_1\ldots$ in $\M$ the sequence $\rank(\rho(v_0)),\rank(\rho(v_1)),\ldots$ satisfies the parity condition, that is, the minimal number that occurs infinitely often in this sequence is even.
We write $\M\models\Amc$ and say that \Amc \emph{accepts} $\M$ in case \Amc has an accepting run
on $\M$. The set of all trees accepted by \Amc is denoted by $\Lmc(\Amc)$. An automaton that is identical to $\A$ except that 
the initial state is $q$ is denoted by
$\A[q_I\mapsfrom q]$. The \emph{size} of an automaton $\A$ is the number of its states and is denoted by $|\A|$. 

An NPTA $\Amc$ is \emph{equivalent} to
a formula $\phi\in\muML$ over a class $\ClassMod$ of trees when $\Mmc\models \varphi$ iff $\Mmc\models
\Amc$ for every tree $\Mmc\in\ClassMod$.  We rely on the following
classical result. 
%
\begin{thm} \label{thm:munpta}
  For every $\muML$-formula $\phi$ and class $\ClassMod$ of trees we can construct an NPTA $\A=(Q,\Sigma,q_I,\delta,\rank)$ with alphabet $\Sigma=\powerset{\sig(\phi)}$, size $|Q|$ exponential in $|\phi|$, and number of ranks $|\{\rank(q)\mid q\in Q\}|$ polynomial in $|\phi|$, such that $\A$ is equivalent to $\phi$ over $\ClassMod$. The construction takes exponential time when $\ClassMod\subset\TT^d$ for some $d$, and doubly exponential time in the unrestricted case. 
\end{thm}
The construction of $\A$ in the above Theorem~\ref{thm:munpta} is classical, and so we only give a brief sketch. For more details see for example~\cite[Theorem~7]{DBLP:conf/icalp/Vardi98}, \cite[Simulation Theorem~7.50]{Venema20}, and the well-presented Dealternation Theorem~14.7 in~\cite{BojanCzerwinski18}. The automaton $\A$ simulating $\phi$ is obtained by composing two simpler automata $\A_1$ and $\A_2$. Given a model $\M$, the first $\A_1$ guesses a positional strategy $\strategy$ for $\eve$ in the semantic game $\game(\M,\phi)$. The other $\A_2$ verifies that this guessed $\strategy$ is winning. The automaton $\A$ is then a product of $\A_1$ and $\A_2$, and has ranks inherited from $\A_2$. The construction of $\A_1$ is straightforward and easily meets the bounds on the number of states. The construction of $\A_2$ with exponentially many states and polynomially many ranks, a key ingredient in the single exponential algorithm for satisfiability of $\muML$~\cite{EmersonJutla88}, follows from the efficient determinization of automata on infinite words~\cite{Safra1988, DBLP:journals/lmcs/Piterman07}.

We will often consider the following \emph{trash normal form} for an NPTA \Amc: there is a distinguished \emph{trash state} $q_\bot$ such that for every state $q$: $\A[q_I\mapsfrom q]$ accepts some model iff $q\neq q_\bot$. Every automaton can be converted to trash normal form by replacing all states whose language is empty with a single trash state. Algorithmically, this boils down to checking emptiness for every state. The latter in turn reduces to solving a suitable parity game, which for $\A=(Q,\Sigma,q_I,\delta,\rank)$ has positions $Q\cup\bigcup_{a\in\Sigma,q\in Q}\delta(q,a)$ and the same number of ranks as $\A$. Since a parity game with $n$ positions and $r$ ranks can be solved in time $n^{O(r)}$ (in fact, even $n^{O(\log(r))}$, see~\cite[Main Theorem 2.7]{quasipolynomialParity2017}), this implies the following complexity bounds.

\begin{lem}\label{lem:trash normal form}
    Let $\A$ be an NPTA over a class $\ClassMod$ of models, with $n$ states and $r$ ranks. Then $\A$ can be brought to trash normal from in time:
    \begin{itemize}
        \item $n^{O(r)}$ if $\ClassMod\subseteq \TT^d$ for some $d$; 
        \item ${(2^n)}^{O(r)}$ in the unrestricted case.
    \end{itemize}
\end{lem}
The above Lemma~\ref{lem:trash normal form} enables the following important strengthening of Theorem~\ref{thm:munpta}.
\begin{cor}\label{cor:muML to NPTA}
  Theorem~\ref{thm:munpta} remains true when we additionally require $\Amc$ to be in trash normal form.
\end{cor}
%


\section{Foundations of Separability}\label{sec:foundations}

We start with recalling the notion of separability and discuss some of its basic
properties. 
%
%
\begin{defi}
  Let $\varphi,\varphi'\in\muML$. A formula $\psi$ \emph{separates} $\varphi$ and $\varphi'$ if $\varphi\models\psi$ and $\psi\models\neg\varphi'$. If $\psi$ comes from a fragment $\LL$ of $\muML$, then we call it an \emph{$\LL$-separator} of $\varphi$ and $\varphi'$. 
\end{defi}
It should be clear that a separator can only exist if $\varphi$ and $\varphi'$ are inconsistent, and that either of $\varphi$ and $\neg\varphi'$ is a trivial separator in this case. Thus, the notion is only interesting if we restrict $\LL$. We illustrate the notion of a separator with a simple example.
\begin{exa}\label{exa:separability}
Consider the $\muML$-formula $\theta_\infty=\nu x.\Diamond x$ expressing the
existence of an infinite path originating in the root, and $\phi'=\Box\bot$
expressing that the root does not have any children. Clearly, $\theta_\infty$ and $\varphi'$ are inconsistent, and for every $n\geq 1$
the modal formula $\Diamond^n\top$ separates $\theta_\infty$ and $\varphi'$.
Hence, there are infinitely many separators. Moreover, since
$\Diamond^{m}\top\models \Diamond^n\top$ whenever $m\geq n$, there does not 
exist a logically strongest separator. \qed
\end{exa}
The main problem we investigate in this paper is \emph{$\ML$-separability of
$\muML$-formulae}, which asks whether given formulae $\phi,\phi'\in\muML$ admit an $\ML$-separator. 
Since any $\ML$-separator of $\phi,\neg\phi$ is equivalent to $\phi$, a natural
special case of $\ML$-separability is \emph{$\ML$-definability}, in which we
assume $\phi'=\neg \phi$.
The notion of separation is closely related to the notion of \emph{interpolation}. Let \Lmc be a set of $\muML$-formulae. A \emph{Craig interpolant} of $\phi,\phi'\in\Lmc$ is an $\Lmc$-separator $\psi$ of $\phi,\neg\phi'$ which only uses the common symbols of $\phi$ and $\phi'$, that is, $\sig(\psi)\subset\sig(\phi)\cap\sig(\phi')$. The \emph{Craig interpolant existence problem} for $\LL$ asks if a given pair of $\LL$-formulae admit a Craig interpolant.
We say that $\Lmc$ enjoys the \emph{Craig interpolation property (CIP)} if every $\phi,\phi'\in\Lmc$ with $\phi\models\phi'$ admit a Craig interpolant. In logics enjoying the CIP, the Craig interpolant existence problem reduces to checking if $\phi\models\phi'$, but in logics which lack CIP the problem may be more involved.

 
All introduced notions can be relativized to a class $\ClassMod$ of models by restricting entailment to the models in $\ClassMod$.
In this paper, we investigate $\ML$-separability and $\ML$-definability of $\muML$-formulae, as well as the Craig interpolant existence problem for $\ML$. The problems are considered over the class of all models and the classes $\TT^d$, $d\in\mathbb{N}$ of trees of bounded outdegree.



A useful observation is that, over all classes $\ClassMod$ considered in this paper, if $\varphi,\varphi'$ admit an $\ML$-separator $\psi$, then they admit an $\ML$-separator over signature $\sigma=\sig(\varphi)\cup\sig(\varphi')$. 
Indeed, consider $\psi_\sigma$ obtained from $\psi$ by replacing every non-$\sigma$ proposition with $\bot$, and $\ClassMod_\sigma$ the subclass of $\ClassMod$ consisting of models in which all non-$\sigma$ propositions are false everywhere. Over $\ClassMod_\sigma$, the formulae $\psi$ and $\psi_\sigma$ are equivalent and hence $\psi_\sigma$ separates $\phi,\phi'$. Since every model in $\ClassMod$ is $\sigma$-bisimilar to a model in $\ClassMod_\sigma$, it follows that $\psi_\sigma$ separates $\phi,\phi'$ over $\ClassMod$ as well. Thus, when looking for $\ML$-separators $\psi$ for $\phi,\phi'$ we will only consider formulae over $\sig(\phi)\cup\sig(\phi')$. Moreover, it will be convenient to use also the notion of an \emph{$\ML_\sigma$-separator} which is just a $\ML$-separator whose signature is contained in $\sigma$.

The starting point for the technical developments in the paper are
model-theoretic characterizations for separability. Similar to what
has been done in the context of 
interpolation, see for example~\cite{AbRobinson}, they are given in
terms of joint consistency, which we
introduce next.  
Let $R$ be a binary relation on some class of models, such as $(\sigma,n)$-isomorphism $\cong_\sigma^n$ or $\ML_\sigma^n$-equivalence $\equiv_{\ML_\sigma^n}$. We call two formulae $\phi,\phi'$ \emph{jointly consistent up to $R$} (in short \emph{jointly $R$-consistent}) if there are models $\M\models\varphi$ and $\M'\models\phi'$ with $R(\M,\M')$. For technical reasons we will sometimes also talk about joint consistency of automata $\A,\A'$ in place of formulae $\phi,\phi'$.
Joint $R$-consistency \emph{over a class} $\ClassMod$ of models is defined by
only looking at models from $\ClassMod$. Clearly, if $R'\subset R$ and
$\ClassMod'\subset\ClassMod$ then joint $R'$-consistency over $\ClassMod'$ implies joint
$R$-consistency over $\ClassMod$.
We use the following standard equivalence:
\begin{align}
  \text{$\phi,\phi'$ are \emph{not} $\ML_\sigma^n$-separable over $\ClassMod$
$\iff$ $\phi,\phi'$ are jointly $\bis_\sigma^n$-consistent over
$\ClassMod$.}\tag{$\mathsf{Base}$}\label{eq:non-sep vs n-bis}
\end{align}
for every $\phi,\phi'\in\muML$, $n\in\NN$, signature $\sigma$, and class $\ClassMod$.
The implication from right to left is immediate. The opposite one follows from the observation that for every $n\in\NN$ and finite $\sigma$ there are only finitely many equivalence classes of $\bis_\sigma^n$, and each such class is fully described with a single modal formula.
Let us illustrate the use of Equivalence~\eqref{eq:non-sep vs n-bis} to show non-separability.
\begin{exa}\label{ex:sep2}
  Let $\phi_1$ and $\phi_2$ be $\muML$-formulae expressing the respective properties $P_1$ and $P_2$ from Example~\ref{ex:sep1}.
Let $\M$ be an infinite path in which every point satisfies $a$, and let $\M_n$ be a finite path of length $n$ in which every point satisfies $a$. Then, for each $n$ the models $\M,\M_n$ witness joint 
$\bis^n_{\{a\}}$-consistency of $\varphi_1,\varphi_2$. 
By Equivalence~\eqref{eq:non-sep vs n-bis} this means that $\phi_1,\phi_2$ are not
$\ML_{\{a\}}^n$-separable
for any $n$, and hence not 
$\ML$-separable at all. \qed
\end{exa}

Equivalence~\eqref{eq:non-sep vs n-bis} is also used to prove the following model-theoretic characterization of Craig interpolant existence over trees of bounded outdegree in terms of joint $\bis_\sigma^n$-consistency, which will be useful later on.
\begin{lem} \label{lem:char-interpolant}
  Let $\varphi,\varphi'\in\ML^n$ and $\sigma=\sig(\phi)\cap\sig(\phi')$. The following are equivalent for every $d\in\NN$:
  \begin{enumerate}

    \item $\varphi,\varphi'$ do not admit a Craig interpolant over $\TT^d$.

    \item $\varphi,\neg\varphi'$ are jointly $\bis_\sigma^n$-consistent over $\TT^d$.

  \end{enumerate}
\end{lem}

\begin{proof}
  By the Equivalence~\eqref{eq:non-sep vs n-bis}, the first item~(1) holds iff $\phi,\neg\phi'$ are jointly $\bis_\sigma^m$-consistent for every $m\in\NN$.
  The implication (1)$\Rightarrow$(2) is thus immediate. For
  the converse implication (2)$\Rightarrow$(1), suppose $\varphi,\neg\varphi'$ are jointly
  $\bis_\sigma^n$-consistent and let $\M\models\varphi,\M'\models
  \neg\varphi'$ with $\M\bis_\sigma^n\M'$ witness this. Since
  $\varphi,\varphi'$ have modal depth at most $n$, we can assume
  without loss of generality that $\M,\M'$ have depth at most~$n$. But then $\M\bis_\sigma\M'$ and so $\M,\M'$ witness joint $\bis_\sigma^m$-consistency of
  $\varphi,\neg\varphi'$
  for every $m$.
\end{proof}

In the remainder of the section, we recall first, in Section~\ref{sec:definability}, the (known) tools that were used to solve $\ML$-definability of $\muML$-formulae, and show then, in Section~\ref{sec:separability} how to extend these tools (a) to solving $\ML$-separability of $\muML$-formulae and (b) to constructing separators and definitions when they exist, both over the class of all models.

\subsection{Modal Definability: A Recap}\label{sec:definability}
In his seminal paper~\cite{DBLP:conf/stacs/Otto99} Otto showed that $\ML$-definability of $\muML$-formulae is \ExpTime-complete over all models and over $\TT^d$ for every $d\geq 2$.

\begin{thmC}[{\cite[Main Theorem and Proposition 3]{DBLP:conf/stacs/Otto99}}]\label{thm:otto main}

  Over the class of all models, as well as over $\TT^d$ for every $d\geq 2$, $\ML$-definability of $\muML$-formulae is \ExpTime-complete.
\end{thmC}
We start by recalling and rephrasing Otto's construction and fixing a small mistake in the original proof. The lower bound follows by an immediate reduction from satisfiability of $\muML$-formulae. We look at the upper bound.
The first step is the following lemma, which is the heart of~\cite[Lemma 2]{DBLP:conf/stacs/Otto99}.
\begin{lem}\label{lem:otto 2}
For every $\phi\in\muML$ and $n,d\in\NN$ the following are equivalent:
\begin{enumerate}
  \item\label{it:otto 1a} $\phi,\neg\phi$ are jointly
    $\bis_\sigma^n$-consistent over $\TT^d$.
  \item\label{it:otto 1b} $\phi,\neg\phi$ are jointly $\cong_\sigma^n$-consistent over $\TT^d$.
\end{enumerate}
\end{lem}
The lemma is true, but its proof in~\cite{DBLP:conf/stacs/Otto99} is
mistaken. The problem there is that the construction duplicates
subtrees and hence may turn $d$-ary models into ones with outdegree
greater than~$d$. We present an easy alternative proof.
\begin{proof}
Only the implication~\ref{it:otto 1a}~$\Rightarrow$~\ref{it:otto 1b} is nontrivial. To prove it assume $d$-ary $\M\models\phi$, $\N\models\neg\phi$ with $\M\bis_\sigma^n\N$ and assume towards contradiction that $\phi,\neg\phi$ are not $\cong_\sigma^n$-consistent over $\TT^d$. We have $\M\cong_\sigma^n\M_{|_{n}}^\sigma\bis\M'$ where $\M^\sigma$ is the $\sigma$-reduct of $\M$, and $\M'\in\TT^d$ is the bisimilarity quotient of its $n$-prefix $\M_{|_{n}}^\sigma$. By the assumption that $\phi,\neg\phi$ are not jointly $\cong_\sigma^n$-consistent, $\M\models\phi$ implies $\M_{|_{n}}^\sigma\models\phi$. By invariance of $\phi$ under $\bis$, this in turn implies $\M'\models\phi$. We construct $\N'\models\neg\phi$ symmetrically. By definition, $\M\bis_\sigma^n\N$ means that $\M_{|_{n}}^\sigma$ and $\N_{|_{n}}^\sigma$ are bisimilar, which is equivalent to saying that their bisimilarity quotients $\M'$ and $\N'$ are isomorphic, and hence $(\sigma,n)$-isomorphic. Thus, $\M',\N'$ witness joint $\cong_\sigma^n$-consistency of $\phi,\neg\phi$ over $\TT^d$, a contradiction.
\end{proof}

Using automata-based techniques we can decide whether
Item~\ref{it:otto 1b} in Lemma~\ref{lem:otto 2} holds for all~$n$.

\begin{prop}\label{prop:automata step}
Let $d\in\NN$ and $\A,\A'$ be parity automata over $d$-ary trees with respective alphabets $\Sigma,\Sigma'$. Then,
$\A,\A'$ are jointly $\cong_\sigma^n$-consistent over $\TT^d$ for all $n\in\NN$ iff $\A,\A'$ are jointly $\cong_\sigma^m$-consistent over $\TT^d$ for $m=|\A|\times|\A'|+1$. If $\A,\A'$ are in trash normal form then the latter condition can be checked in time polynomial in $(|\Sigma|+|\Sigma'|)\times(|\A|+|\A'|)^d$. 
\end{prop}
\begin{proof}(Sketch)
Without loss of generality both $\A$ and $\A'$ work over the same alphabet, $\Sigma=\Sigma'=\powerset{\sigma_\textsf{aut}}$ for some signature $\sigma_\textsf{aut}$, and $\sigma\subset\sigma_\mathsf{aut}$.
Due to well-known relativization techniques we also do not lose generality by only running $\A,\A'$ on full $d$-ary trees with no leaves. Let $L$ be the language of finite full $d$-ary trees over $\sigma$ such that $\M\in L$ iff $\M$ is a prefix of a $\sigma$-reduct of a model of $\A$. Let $L'$ be the analogous language for $\A'$. 
The \emph{tallness} of a finite tree is the minimal distance from the root to a leaf. Observe that $\A,\A'$ are $\cong_\sigma^n$-consistent over $\TT^d$ iff $L\cap L'$ contains a tree of tallness $n$. Thus, it suffices to check if $L\cap L'$ contains trees of arbitrarily high tallness. 

To that end construct an automaton $\B_0$ recognizing $L$.
Denote $\A=(Q,\Sigma,\delta,q_I,\rank)$. Assume that $\A$ is in trash normal form (otherwise transform it to that form), and let $q_\bot$ denote the trash state. Then $\B_0=(Q,\powerset{\sigma},\delta',q_I,\rank_\bot)$, where $\rank_\bot(q)=1$ for all $q$ so that no infinite run is accepting, and
\[
\delta'(q,a) = 
\begin{cases}
 \{\emptyset\}\cup\bigcup\limits_{\substack{b\in\Sigma\\b\cap\sigma=a}}\delta(q,b) & \text{if $q\neq q_\bot$,}\\
 \emptyset & \text{otherwise;}
\end{cases}
\]
for all $a\in\powerset{\sigma},q\in Q$. Note that transitions from $q$ to $\emptyset$ in the above $\delta'$ enable $q$ to be assigned to a leaf in an accepting run of $\B_0$.
Clearly, if $\A$ is in the trash normal form then $\B_0$ is computed in time polynomial in $|\Amc|$ and $|\Sigma|$.
Let $\B_0'$ be an analogous automaton for $L'$. Construct from $\B_0,\B_0'$ an automaton $\B$ recognizing $L\cap L'$, of size $|\B|=|\B_0|\times|\B_0'|=|\A|\times|\A'|$. An easy pumping argument shows that the language $L\cap L'$ of $\B$ contains trees of arbitrarily high tallness iff it contains a tree of tallness $m=|\B|+1$. To test the latter condition it is enough to inductively compute a sequence $S_1\supset S_2\supset...\supset S_{|\B|+1}$ of subsets of states of $\B$, where $S_i$ is the set of all states $q$ such that $\B[q_I\mapsfrom q]$ recognizes a tree of tallness at least $i$.
\end{proof}

We are ready to solve $\ML$-definability over $\TT^d$ in exponential time. Assume $\muML$-formula $\phi$ over signature $\sigma$. For every $n$, we know by Equivalence~\eqref{eq:non-sep vs n-bis} that $\phi$ is equivalent over $\TT^d$ to some $\psi\in\ML_\sigma^n$ iff $\phi,\neg\phi$ are not jointly $\bis_\sigma^n$-consistent over $\TT^d$. By Lemma~\ref{lem:otto 2} this is equivalent to the lack of joint $\cong_\sigma^n$-consistency of $\phi,\neg\phi$ over $\TT^d$. By Corollary~\ref{cor:muML to NPTA} we can compute exponentially-sized automata $\A$, $\A'$ in trash normal form, equivalent to $\phi$ and $\neg\phi$ over $\TT^d$. It follows that $\phi$ is not $\ML$-definable over $\TT^d$ iff $\A,\A'$ are joint $\cong_\sigma^n$-consistent over $\TT^d$ for every $n$. The last condition is decided using Proposition~\ref{prop:automata step}.
The runtime of our algorithm is polynomial in $(|\Sigma|+|\Sigma'|)\times(|\A|+|\A'|)^d$, and thus exponential in $|\phi|$. This proves the part of Theorem~\ref{thm:otto main} about $\TT^d$. The remaining part concerning unrestricted models is a special case of Theorem~\ref{thm:separability unrestricted}, which we will prove in the following subsection.


%

\subsection{Modal Separation: the Unrestricted Case}\label{sec:separability}
Over unrestricted models, separability turns out to be only slightly more complicated than definability. Lemma~\ref{lem:otto 2} becomes false if $\neg\phi$ is replaced with arbitrary $\phi'$ (which would be the statement relevant for separability). We have the following lemma, however.

\begin{lem}\label{lem:unbounded separability}
For every $\phi,\phi'\in\muML$ and $n\in\NN$ the following are equivalent:
\begin{enumerate}
  \item\label{it:unbounded separability 1} $\phi,\phi'$ are jointly $\bis_\sigma^n$-consistent over all models.
  \item\label{it:unbounded separability 2} $\phi,\phi'$ are jointly $\cong_\sigma^n$-consistent over $\TT^d$, where $d=|\phi|+|\phi'|$.
\end{enumerate}
\end{lem}
\begin{proof}
  The implication \eqref{it:unbounded separability 1}$\Leftarrow$\eqref{it:unbounded separability 2} is immediate. To prove the other one \eqref{it:unbounded separability 1}$\Rightarrow$\eqref{it:unbounded separability 2} consider an intermediate property:
  \begin{align}
  \text{$\phi,\phi'$ are jointly $\cong_\sigma^n$-consistent over all models.}\tag{1.5}\label{eq:unbounded sep intermediate}
  \end{align}
  The implication \eqref{it:unbounded separability 1}$\Rightarrow$\eqref{eq:unbounded sep intermediate} can be read off from Otto's original proof. The remaining one \eqref{eq:unbounded sep intermediate}$\Rightarrow$\eqref{it:unbounded separability 2} is a special case of a stronger claim which we prove later: the implication \eqref{it:graded model theory 3}$\Rightarrow$\eqref{it:graded model theory 4} of Lemma~\ref{lem:graded model theory}.
\end{proof}
Lemma~\ref{lem:unbounded separability} allows us to solve $\ML$-separability in exponential time. 

\begin{thm}\label{thm:separability unrestricted}
  Over all models, $\ML$-separability of $\muML$-formulae is \ExpTime-complete.
\end{thm}
\begin{proof}
The proof is almost the same as our proof of Theorem~\ref{thm:otto main}. The only difference is that we consider an arbitrary $\phi'$ in place of $\neg\phi$, and hence use Lemma~\ref{lem:unbounded separability} in place of Lemma~\ref{lem:otto 2}.
\end{proof}
Apart from deciding separability we also construct separators when they exist. Given a subset $\LL$ of $\muML$-formulae, $\phi\in\muML$, we call $\psi\in\LL$ an \emph{$\LL$-uniform consequence} of $\phi$ if
\[
\psi\models\theta \hspace{0.5cm} \iff \hspace{0.5cm} \phi\models\theta
\]
for every $\theta\in\LL$.
Note that the above condition implies $\phi\models\psi$.
The notion relativizes to a fixed class $\ClassMod$ of models by only considering entailment over that class. 
Observe that if $\phi,\phi'$ are $\LL$-separable and $\psi\in\LL$ is an $\LL$-uniform consequence of $\phi$ then $\psi$ is an $\LL$-separator for $\phi,\phi'$. The same is true over any class $\ClassMod$.


Note that it follows from the proof of Theorem~\ref{thm:separability unrestricted} that if $\phi,\phi'$ are $\ML$-separable then they admit a separator of modal depth $n$ at most exponential in $|\phi|+|\phi'|$.
It follows that constructing an $\ML$-separator for $\phi,\phi'$ boils down to constructing an $\ML_\sigma^n$-uniform consequence of $\phi$ where $\sigma=\sig(\phi)\cup\sig(\phi')$. A naive construction which always works is to take the disjunction of all $\ML_\sigma^n$-types consistent with $\phi$ over $\ClassMod$. Here, by an \emph{$\ML_\sigma^n$-type} we mean a maximal consistent subset of $\ML_\sigma^n$. Since up to equivalence there are only finitely many formulae in $\ML_\sigma^n$, each
$\ML_\sigma^n$-type can be represented as a single $\ML_\sigma^n$-formula and the
mentioned disjunction $\psi$ is well-defined. This construction is
non-elementary in $n$ over all models and doubly exponential in $n$
over models of bounded outdegree.

For some classes $\ClassMod$ a (i) stronger and (ii) more succinct construction is possible. Given $\phi\in\muML$ and $n\in\NN$ we build $\psi\in\ML^n$ which (i) is an $\ML^n$-uniform (rather than $\ML_\sigma^n$-uniform) consequence of $\phi$, and (ii) whose dependence on $n$ is only exponential.
This better construction works over unrestricted models, over $\TT^1$, and over $\TT^2$, but notably not over $\TT^d$ with $d\geq 3$. 
The following example shows that this gap is unavoidable, as already over $\TT^3$, $\ML^n$-uniform consequences need not exist.
Recall the formula $\theta_\infty$ from Example~\ref{exa:separability} and consider $\phi=\Diamond(\theta_\infty\wedge a)\wedge\Diamond(\neg\theta_\infty\wedge a)$ which says that there are at least two children satisfying $a$: one with an infinite path and the other without. Over ternary models $\phi$ enforces that there is at most one child satisfying $\neg a$, and hence entails $\theta_c=\Diamond(\neg a\wedge c)\Rightarrow\Box(\neg a\Rightarrow c)$ for every $c$. It is easy to check that over $\TT^3$ any consequence $\psi\in\ML$ of $\phi$ can only entail $\theta_c$ when $c\in\sig(\psi)$. Consequently, $\phi$ has no $\ML^1$-uniform consequence over $\TT^3$. 

We now go back to the construction of $\ML^n$-uniform consequences in the feasible cases: all models, $\TT^1$, and $\TT^2.$
Since in the following Section~\ref{sec:unary} we will provide a more efficient construction for $\TT^1$, we focus on the unrestricted and binary case. For convenience, we construct $\ML^n$-uniform consequences of \emph{automata} instead of formulae, with definitions adapted in an obvious way. In the binary case we will assume that automata are \emph{duplication safe} meaning that for every singleton transition $(p)\in\delta(q,c)$ its duplication is also legal, that is,  $(p,p)\in\delta(q,c)$. Every automaton equivalent to a $\muML$-formula can be made duplication safe.

\begin{prop}\label{prop:L-uniform construction}
Let $\ClassMod$ be the class of all models or $\TT^2$, and let \Amc be an NPTA over $\ClassMod$ with alphabet $\Sigma$. Assume that \Amc is in trash normal form, and in the case $\ClassMod=\TT^2$ also that it is duplication safe. For every $n\in\NN$, an $\ML^n$-uniform consequence of $\Amc$ over $\ClassMod$ can be constructed in time
${(|\Sigma|\times2^{|\Amc|})^{O(n)}}$ 
if $\ClassMod$ is the class of all models and in time $(|\Sigma|\times|\Amc|)^{O(n)}$ 
if $\ClassMod=\TT^2$. 
\end{prop}

\begin{proof}
Let $\A=(Q,\Sigma,q_I,\delta,\rank)$ be an NPTA over $\ClassMod$ in trash normal form. In the case $\ClassMod=\TT^2$ assume that $\A$ is also duplication safe.
We construct $\psi_{m,q}$ for every $q\in Q$ and $m\leq n$ by induction on $m$. For the base case we put:
\[
  \psi_{0,q} = \bigvee\{c\in\alphabet\ |\ \text{there is $\N\in\ClassMod$ with $\N\models\Amc[q_I\mapsfrom q]$ and $\N\models c$}\}
\]
For the induction step define:
\[
  \psi_{m+1,q} = \bigvee_{c\in\alphabet} \bigvee_{S\in\delta(q,c)}c \wedge \nabla \{\psi_{m,p}\ |\ p\in S\}
\]
where $\nabla\Phi$ is an abbreviation for $\bigwedge_{\theta\in\Phi}\Diamond\theta\wedge\Box\bigvee_{\theta\in\Phi}\theta$. Assume $\ClassMod$ is either the class of all models or $\TT^2$.
The construction preserves the following invariant:
\begin{align}
  \M\models\psi_{m,q} \text{\hspace*{0.5cm} $\iff$ \hspace*{0.5cm} there exists $\N\in\ClassMod$ with $\N\models\Amc[q_I\mapsfrom q]$ and $\M\bis^m\N$}\label{eq:universal consequence}
\end{align}
for every structure $\M\in\ClassMod$. 
Hence, $\psi_{n,q_I}$ is an $\ML^n$-uniform consequence of $\A$ over $\ClassMod$. Thanks to the trash normal form of \Amc, the base case formula is computed in polynomial time. The inductive case consists of a simple syntactic manipulation. Thus, the claimed complexity follows from the fact that 
$\psi_{q,n}$ has the right size for all $q$ and $n$, which can be checked by a routine induction on $n$.

The proof of~\eqref{eq:universal consequence} proceeds by induction. The base case is immediate so we focus on the inductive case. We start with the easier implication $\Leftarrow$ from right to left. Assume $\M\bis^{m+1}\N$ for some $\N\in\ClassMod$ such that $\N\models\Amc[q_I\mapsfrom q]$. Let $c$ and $N_0$ denote respectively the color and the set of all children of the root of $\N$. There is an accepting run $\rho$ of $\Amc[q_I\mapsfrom q]$ on $\N$. We have $\rho[N_0]=S$ for some transition $S\in\delta(q,c)$. For each $\point\in N_0$ the run $\rho$ witnesses that the subtree of $\N$ rooted at $\point$ is accepted by $\Amc[q_I\mapsfrom\rho(\point)]$. By induction hypothesis, this means that the subtree satisfies $\psi_{\rho(\point),m}$. It follows that $\N$ satisfies $\nabla \{\psi_{m,p}\ |\ p\in S\}$ and hence also $\psi_{q,m+1}$. By $\M\bis^{m+1}\N$ this implies $\M\models\psi_{q,m+1}$.

We now prove the other implication $\Rightarrow$ from left to right. The proof uses similar but slightly different arguments in the cases of binary and unrestricted models. It is worth to point out, however, that the implication $\Rightarrow$ would not be valid over $\TT^d$ with $d\geq 3$.
Assume $\M\models\psi_{m+1,q}$, let $c$ be the color of the root and $M_0$ its children. There is $S\in\delta(q,c)$ such that $\nabla(\Phi)$ where $\Phi=\{\psi_{m,p}\ |\ p\in S\}$ is satisfied in the root.
This means that the subtree of each $\point\in M_0$ satisfies some $\psi_{m,p}$, and conversely each $\psi_{m,p}$ is satisfied in some of these subtrees. It is worth noting that $\Phi$ contains at most two elements in case $\ClassMod=\TT^2$.

We define a model $\N'$. Its exact definition depends on which class $\ClassMod$ of models we consider. 
If $\ClassMod$ is the class of all models we obtain $\N'$ from $\M$ by replacing the subtree of each $\point\in M_0$ with $|Q|$ distinct copies of that subtree. If $\ClassMod=\TT^2$ then $\N'$ is the full binary tree obtained from $\M$ by duplicating subtrees.
In either case, let $N_0'$ denote the children of the root of $\N'$.
By design $\N'$ has a separate witness for each formula in $\Phi$.
That is, there is a surjective assignment $h:N_0'\to\Phi$ which maps every $\point\in N_0'$ to some formula $\psi_{m,p}$ true in the subtree of $\point$.

By the induction hypothesis, for each $\point\in N_0'$ with $h(\point)=\psi_{m,p}$ there is a model $\N_p\models\Amc[q_I\mapsfrom p]$ $m$-bisimilar to the subtree of $\point$. Define $\N$ as follows: first take the disjoint union $\{\point\}\uplus\biguplus\{\N_p\ |\ p\in S\}$ of all the $\N_p$'s and a fresh point $v$ of color $c$; then for every $\N_p$ add an edge from $\point$ to the root of $\N_p$ and set $\point$ as the new root. It is easy to see that $\M\bis\N'\bis^{m+1}\N$ and $\N\models\Amc[q_I\mapsfrom q]$ (in the case $\ClassMod=\TT^2$ the latter may use the assumption that $\Amc$ is duplication safe). This completes the proof of~\eqref{eq:universal consequence}, and hence proves Proposition~\ref{prop:L-uniform construction}.
\end{proof}

Given the exponential construction of automata from
Corollary~\ref{cor:muML to NPTA} and the exponential upper bound on
modal depth $n$ of separators, Proposition~\ref{prop:L-uniform
construction} yields an efficient construction of separators.
%
\begin{thm}\label{thm:construction unrestricted}
  If $\phi,\phi'$ are $\ML$-separable, then one can compute an $\ML$-separator in time doubly exponential in $|\phi|+|\phi'|$.
  %
\end{thm}
It is not difficult to show that, in the presence of two
modalities $\Diamond_1,\Diamond_2$ over independent accessibility relations $\to_1,\to_2$, the construction is
optimal: one can express by a $\muML$-formula of size polynomial in $n$ that the model embeds a full binary tree of depth $2^n$ in which each
inner node has both a $\to_1$- and a $\to_2$-successor.
Using standard techniques, one can show that any modal formula
expressing this property is of doubly exponential size~\cite{French13}. Intuitively, for every inner node, the modal formula needs to have a subformula inside a sequence of modalities that addresses this node. Since there are doubly exponentially many possible addresses in the stipulated tree, a doubly exponential lower bound on formula size follows. Whether having two accessibility relations is necessary for this lower bound is an interesting question which we leave open. Note that the obvious idea of simulating the two accessibility relations by a single one and using an auxiliary proposition $p$ (that is, $\Diamond_1\psi$ is replaced by $\Diamond (p\wedge\psi)$ and $\Diamond_2\psi$ is replaced by $\Diamond (\neg p\wedge\psi)$) does not work here, since we can simply enforce such embedded tree using~$\bigwedge_{i<2^n}\Box^i(\Diamond p\wedge \Diamond \neg p)$.

The remaining open cases are the problems of $\ML$-separability (and separator construction) over $\TT^d$ for $d\geq 1$. We investigate the cases of unary ($d=1$), binary ($d=2$), and higher maximal outdegree ($d\geq 3$) in turn, starting in Section~\ref{sec:unary}. 
We emphasize that the outdegree $d$ is not a part of the input but rather a property of the considered class of models.

\section{Unary Case}\label{sec:unary}

We first investigate $\ML$-separability over $\TT^1$, that is, models
that are essentially finite or infinite \emph{words}. Note that satisfiability of $\muML$
over words is \PSpace-complete: an upper bound follows, e.g.~via the
translation to automata and the lower bound is inherited from
$\LTL$~\cite[Theorem 4.1]{DBLP:journals/jacm/SistlaC85}. This suggests that also
definability and separability could be easier. Indeed, we show: 
\begin{thm}
  $\ML$-definability and $\ML$-separability of $\muML$-formulae are
  \PSpace-complete over $\TT^1$.
\end{thm}
\begin{proof}
The lower bound is by a reduction from satisfiability, and
applies to definability. Indeed, let $\phi\in\muML$ and let $\psi=\mu x.(a\vee \Diamond x)$ be the $\muML$-formula expressing that $a$ is satisfied at some point in the model, for a proposition $a$ not occurring in $\phi$. It is easy to see that $\phi$ is unsatisfiable over $\TT^1$ iff $\phi\wedge \psi$ is $\ML$-definable (namely $\bot$).

For the upper bound, consider formulae $\phi,\phi'\in\muML$ and a signature $\sigma$. Note that two models in $\TT^1$ are bisimilar iff they are isomorphic.
Hence, by Equivalence~\eqref{eq:non-sep vs n-bis}, $\phi,\phi'\in\muML$ are not $\ML_\sigma$-separable over $\TT^1$ iff they are jointly $\cong_\sigma^n$-consistent for all $n$. By Proposition~\ref{prop:automata step} and Corollary~\ref{cor:muML to NPTA}, we get a constant $m$ exponential in $|\phi|+|\phi'|$ such that $\phi,\phi'$ are jointly $\cong_\sigma^n$-consistent over $\TT^1$ for all $n$ iff they are jointly $\cong_\sigma^m$-consistent over $\TT^1$. We therefore decide the latter. We do that by a reduction to satisfiability of $\muML$ over $\TT^1$, which is known to be in \PSpace~\cite[Corollary 4.6]{DBLP:conf/popl/Vardi88}.

Let $\psi,\psi'$ be respective copies of $\phi,\phi'$ in which every symbol was replaced with a fresh copy, so that $\psi$ and $\psi'$ have disjoint signatures. Observe that $\phi,\phi'$ are jointly $\cong_\sigma^m$-consistent over $\TT^1$ iff there are models $\M\models\psi$ and $\M'\models\psi'$ in $\TT^1$ which agree on copies of all propositions from $\sigma$, up to position $m$. That is, every position $i\leq m$ satisfies $\theta_\mathsf{agree}=\bigwedge_{a\in\sigma}(a_\psi\Leftrightarrow a_{\psi'})$, where $a_\psi$ and $a_{\psi'}$ denote the respective copies of $a$ in $\psi$ and $\psi'$. The requirement that $\theta_\mathsf{agree}$ holds in all points up to depth $m$ can be enforced by a formula $\theta$ which uses $k=\lceil\log(m)\rceil$ fresh auxiliary variables to implement a binary counter in a standard way. This completes the reduction, as $\psi\wedge\psi'\wedge\theta$ is satisfiable in $\TT^1$ iff $\phi,\phi'$ are not $\ML_\sigma$-separable over $\TT^1$. Since $k$ is polynomial in $|\phi|+|\phi'|$, this completes the proof.
\end{proof}

We conclude this section with proving that $\ML$-separators can be constructed in
exponential time and are thus of at most exponential size. 
Note that this is optimal, since 
over~$\TT^1$, $\muML$ is exponentially more succinct than $\ML$.
Indeed, it is standard to implement an exponential counter using a
polynomially sized $\muML$-formula.
\begin{thm} \label{thm:word-construction} 
  If $\varphi,\varphi'\in\muML$ are $\ML$-separable over
  $\TT^1$, then one can compute an $\ML$-separator 
  in time exponential in $|\varphi|+|\varphi'|$. 
\end{thm}
As argued in the previous section, it suffices to construct
an $\ML^n$-uniform consequence of the NPTA equivalent to
$\varphi$, which we do next. 
\begin{restatable}{prop}{propluniformone}\label{prop:L-uniform construction1}
  Let $\A$ be an NPTA over $\TT^1$ with alphabet $\Sigma$.
  For every $n\in\NN$, an $\ML^n$-uniform consequence of $\A$ over $\TT^1$ can be
  constructed in time polynomial in $n$, $|\Sigma|$, and~$|\Amc|$.
\end{restatable}
\begin{proof} 
  Take  an NPTA $\Amc=(Q,\Sigma,\delta,q_I,\rank)$ and $n\in \NN$.
  We first construct, for every $m\leq n$ and $p,q\in Q$, formulae $\psi^m_{pq}\in
  \ML_\sigma^n$, such that for every $\M\in\TT^1$: 
\begin{equation}
  \M\models \psi^m_{pq}\text{\quad iff\quad there is a run of \Amc from $p$ to $q$ on the $m$-prefix of $\M$.}\label{eq:word-run}
\end{equation}
The definition is by induction on $m$:
\begin{align*}
  \psi^0_{pq} & = \text{ if $p\neq q$ then $\bot$ else $\top$} \\
  \psi^1_{pq} & = \bigvee\{c\ \mid\ c\in\alphabet, \{q\}\in \delta(p,c)\} \\
  \psi^{m}_{pq} & = \bigvee_{q'\in Q} \big(\psi^{\lfloor m/2\rfloor
  }_{pq'}
  \wedge\Diamond^{\lfloor m/2\rfloor}\psi^{\lceil m/2\rceil
  }_{q'q}\big)\quad \text{for $1<m\leq n$}
\end{align*}
It is routine to verify that $\psi^m_{pq}$ satisfies~\eqref{eq:word-run} and is
of size $|\psi^m_{pq}|\in O(|Q|\cdot m^2)$.

In the construction of the $n$-uniform consequence of \Amc, a bit of care has to be taken since \Amc may accept words shorter or longer than $n$. To deal with this, we introduce some more notation. Since we are working over
$\TT^1$, $\delta(q,c)$
contains only sets of cardinality at most $1$. The 
case $\emptyset\in \delta(q,c)$ is of particular interest because
this means that the automaton can accept in state $q$ reading color $c$. Denote with
$\textit{Acc}_q$ the set of $c\in\Sigma$ with $\emptyset\in
\delta(q,c)$. Further denote with $\textit{Cont}_q$ the set of all
$c$ such that $\Amc[q_I\mapsfrom q]$ accepts a word starting with $c$.
We finish the construction by setting:
\[\psi_n = \bigvee_{q\in Q} \left(\psi^n_{q_Iq}\wedge 
  \Box^n\bigvee_{c\in \textit{Cont}_q} c\right)\vee
  \bigvee_{m\leq n}\bigvee_{q\in Q} \left(\psi^m_{q_Iq}\wedge \Box^{m+1}\bot \wedge
\Box^{m}\bigvee_{c\in \textit{Acc}_q}c\right).\]
Intuitively, $\psi_n$ describes all possible prefixes of length $n$ of words accepted by $\Amc$. 
It is readily checked that $\psi_n$ satisfies the required size bounds and that $\Amc\models\psi_n$.
To verify that $\psi_n\models \theta$ for every $\theta\in \ML^n$ with
$\Amc\models \theta$, we show the following equivalence for all
$\Mmc\in\TT^1$:
\begin{align}
  \M\models\psi_n \text{\hspace*{0.5cm} $\iff$ \hspace*{0.5cm} there exists
  $\N\models \Amc$ with $\N\leftrightarroweq^n \M$.}
  \label{eq:word-universal-consequence} 
\end{align} 

For ``$\Rightarrow$'', fix $\M\in\TT^1$ with $\M\models
\psi_n$. If $\M\models \psi^n_{q_Iq}\wedge
  \Box^n\bigvee_{c\in\textit{Cont}_q}c$ for some $q$, then by Invariant~\eqref{eq:word-run},
there is a run of $\Amc$ from the initial state $q_0$ to some state
$q\in Q$ when reading the $n$-prefix of $\M$, and the last color in
the prefix is $c$. Since $c\in \textit{Cont}_q$, we can extend the $n$-prefix of $\M$ to
an $\N\in \TT^1$ accepted by \Amc. Clearly, in this way $\N\bis^n\M$. If $\M\models
\psi^m_{q_Iq}\wedge\Box^{m+1}\bot\wedge \Box^m\bigvee_{c\in
\textit{Acc}_q}c$, for some $m\leq n$ and $q\in Q$, then $\M$ is a
finite word
of length $m$ that is accepted by the automaton. We can take
$\N=\M$ in this case. 

For ``$\Leftarrow$'', let $\M\in \TT^1$ a word such that there is some
$\N\models\Amc$ with $\N\bis^n \M$. The former condition $\N\models\Amc$
implies the existence of an accepting
run $\rho$ of $\Amc$ on
$\N$. The latter, $\N\bis^n \M$,
implies that $\N$ and $\M$ coincide on their
$n$-prefixes. We distinguish cases. If the length of $\N$ is greater than
$n$, then the $n$-prefix of $\rho$ ending in state $q$ witnesses 
$\M\models \psi^n_{q_Iq}\wedge \Box^n\bigvee_{c\in\textit{Cont}_q}c$.
Otherwise, the depth of $\N$ is $m\leq n$ and the run $\rho$ ending in
$q$ witnesses that $\M\models
\psi^m_{q_Iq}\wedge\Box^{m+1}\bot\wedge \Box^m\bigvee_{c\in
\textit{Acc}_q}c$.
\end{proof}

\section{Binary Case}\label{sec:binary}

We next handle the binary case $\TT^2$. The key observation here is that, between full binary trees, bisimilarity entails isomorphism.
\begin{prop}\label{prop:bin bis=iso}
Assume full binary trees $\M,\M'\in\TT^2$. If $\M$ and $\M'$ are $\sigma$-bisimilar then they are $\sigma$-isomorphic.
\end{prop}
\begin{proof}
By definition a $\sigma$-bisimulation between two models is a bisimulation between their reducts to $\sigma$, and $\sigma$-isomorphism is such a bisimulation which is additionally bijective. It therefore suffices to show that if $\M,\M'$ are full binary trees and $Z$ is a bisimulation between them then
there is a bijective bisimulation $Z'\subset Z$. We construct such a $Z'$. To that end, we inductively construct a descending sequence:
\[
  Z \supset Z_1\supset Z_2\supset Z_3\supset...
\]
of bisimulations such that for each $n$ the restriction of $Z_n$ to the $n$-prefixes of $\M$ and $\M'$ is bijective. 
The induction base $n=1$ is trivial with $Z_1=Z$.
For the induction step $n+1$ let $Z_n\subset Z$ be the bisimulation given by the inductive hypothesis. $Z_n$ bijectively maps the points $\point_1,...,\point_k$ at depth $n$ in $\M$ to the respective points $\altpoint_1,...,\altpoint_k$ at depth $n$ in $\M'$. For each $i$ we have $\point_i Z \altpoint_i$. Hence, either both $\point_i$ and $\altpoint_i$ are leaves (i), or both have respective children $\point_i^l,\point_i^r$ and $\altpoint_i^l,\altpoint_i^r$. In the latter case either (ii) $\point_i^l Z\altpoint_i^l$ and $\point_i^r Z\altpoint_i^r$ or (iii) $\point_i^l Z\altpoint_i^r$ and $\point_i^r Z\altpoint_i^l$ (the cases (ii) and (iii) are not exclusive). Consider the bijective relation $K^i\subset Z$ between children of $\point_i$ and children of $\altpoint_i$:
\begin{align*}
  K^i =
  \begin{cases}
    \emptyset & \text{if $\point_i$ and $\altpoint_i$ are leaves (i),}\\
    \{(\point_i^l,\altpoint_i^l), (\point_i^r,\altpoint_i^r)\} & \text{if (ii),}\\
    \{(\point_i^l,\altpoint_i^r), (\point_i^r,\altpoint_i^l)\} & \text{otherwise.}
  \end{cases}
\end{align*}
The bisimulation $Z_{n+1}$ is constructed as follows. It is identical to $Z_n$ between points at levels at most $n$, to $\bigcup_{i\leq k}K^i$ between points at level exactly $n+1$, and to $Z$ between points at strictly greater levels. No points at mismatching levels are linked. It is straightforward to verify that such $Z_{n+1}\subset Z_n$ is a bisimulation, and that its restriction to the $n+1$-prefixes of $\M$ and $\M'$ is bijective.

We conclude the proof by taking the limit $Z'=\bigcap_{n\in\NN}Z_n$ as the desired bijective bisimulation between $\M$ and $\M'$.
\end{proof}
Proposition~\ref{prop:bin bis=iso} can be used to prove the Craig interpolation
property of $\ML$ over $\TT^2$.

\begin{prop}\label{prop:Craig interpolation bin}
$\ML$ over $\TT^2$ has the Craig interpolation property.
\end{prop}
\begin{proof}
Assume $\phi,\phi'\in\ML$, let $n$ be the maximum of their modal depths and denote $\sigma=\sig(\phi)\cap\sig(\phi')$. Assume towards contradiction that $\phi\models\phi'$ but $\phi,\phi'$ do not have a Craig interpolant. By Lemma~\ref{lem:char-interpolant} the latter implies that $\phi,\neg\phi'$ are jointly $\bis_\sigma^n$-consistent over $\TT^2$. Let $\M\bis_\sigma^n\M'$ be respective binary trees witnessing this. By duplicating subtrees and trimming if needed, we may assume that $\M,\M'$ are full and have depth at most $n$, which therefore by Proposition~\ref{prop:bin bis=iso} implies $\M\cong_\sigma\M'$. Let $\N$ be a model obtained from $\M,\M'$ by taking their common underlying graph, and valuation of every proposition $\atProp\in\sig(\phi)$ or $\atProp\in\sig(\phi')$ inherited from $\M$ or $\M'$, respectively (by $\M\cong_\sigma\M'$ this is well-defined for $\atProp\in\sig(\phi)\cap\sig(\phi')$). It follows that $\M\cong_{\sig(\phi)}\N\cong_{\sig(\phi')}\M'$ and hence $\N$ satisfies both $\phi$ and $\neg\phi'$, a contradiction.
\end{proof}

Apart from the above Proposition~\ref{prop:Craig interpolation bin}, Proposition~\ref{prop:bin bis=iso} also implies the following separability-variant of Lemma~\ref{lem:otto 2} over $\TT^2$.
\begin{lem}\label{lem:bin model theory}
For every $\phi,\phi'\in\muML$ and $n\in\NN$ the following are equivalent:
\begin{enumerate}
  \item\label{it:bin model theory 1} $\phi,\phi'$ are jointly
    $\bis_\sigma^n$-consistent over $\TT^2$.
  \item\label{it:bin model theory 2} $\phi,\phi'$ are jointly $\cong_\sigma^n$-consistent over $\TT^2$.
\end{enumerate}
\end{lem}
\begin{proof}
We show only the nontrivial implication~\ref{it:bin model theory 1}~$\Rightarrow$~\ref{it:bin model theory 2}. Assume binary $\M\models\phi$, $\M'\models\phi'$ with $\M\bis_\sigma^n\M'$. Let $\N\models\phi$ and $\N'\models\phi'$ be full binary trees obtained from $\M$ and $\M'$ by duplicating subtrees. By Proposition~\ref{prop:bin bis=iso}, $\N\cong_\sigma^n\N'$ which proves~\ref{it:bin model theory 2}.
\end{proof}
Similarly to the definability case, Lemma~\ref{lem:bin model theory}
combined with Equivalence~\eqref{eq:non-sep vs n-bis} and
Proposition~\ref{prop:automata step} immediately gives an exponential
procedure for separability. Since the lower bound is inherited from
definability, we get the following result.
\begin{thm}
$\ML$-separability and $\ML$-definability of $\muML$-formulae are \ExpTime-complete over
$\TT^2$.
\end{thm}
With the same argument as for Theorem~\ref{thm:construction unrestricted} we use Proposition~\ref{prop:L-uniform construction} to conclude:
\begin{thm}
  If $\phi,\phi'$ are $\ML$-separable over $\TT^2$, then one
  can compute an $\ML$-separator in time doubly exponential in $|\phi|+|\phi'|$.
\end{thm}

\section{Ternary and Beyond}\label{sec:ternary}

In this section we address the case of models with outdegree
bounded by a number $d\geq 3$. We illustrate that this case behaves differently as it lacks the Craig interpolation property. 

\begin{thm}\label{thm:nocraig}
  For $d\geq 3$, $\ML$ over $\TT^d$ does not enjoy the Craig interpolation property.
\end{thm}

%
\begin{proof}
  We give the proof for $d=3$, the proof for $d>3$ is similar.
  Consider $\ML$-formulae $\varphi = \Diamond(a\wedge b)\wedge
  \Diamond (a\wedge \neg b)$ and $\varphi'=\Diamond(\neg a\wedge
  c)\wedge \Diamond (\neg a\wedge \neg c)$. Clearly,
  $\varphi\models\neg \varphi'$ over $\TT^3$. 
  Observe that models $\M,\M'$ in Figure~\ref{fig:witnesses} witness
  that $\varphi,\varphi'$ are jointly $\bis_{\{a\}}$-consistent and thus
  jointly $\bis_{\{a\}}^n$-consistent
  for every $n\in \NN$. By Equivalence~\eqref{eq:non-sep vs n-bis} 
  there is no $\ML_{\{a\}}$-separator, which is nothing else than a
  Craig interpolant.
\end{proof}

\begin{figure}[t]
  \centering
  \begin{tikzpicture} \tikzset{
			dot/.style = {draw, fill=black, circle, inner
			sep=0pt, outer sep=1pt, minimum size=3pt},
			wdot/.style = {draw, fill=white, circle, inner
			sep=0pt, outer sep=1pt, minimum size=3pt}
		}
		

\draw (0,0.5) node[label=$\M$] {};
\draw (3,0.5) node[label=$\M'$] {};


\draw (0,0) node[dot, red, label=north:$\point_I$] (w) {};
\draw (-1,-.5) node[dot, green, label=south:{$a,b$}] (w1) {};
\draw (0,-.5) node[dot, green, label=south:{$a,\neg b$}] (w2) {};
\draw (1,-.5) node[dot, blue, label=south:{$\neg a$}] (w3) {};

\draw[->] (w)--(w1);
\draw[->] (w)--(w2);
\draw[->] (w)--(w3);

\draw (0+3,0) node[dot, red, label=north:$\point_I'$] (v) {};
\draw (-1+3,-.5) node[dot, blue, label=south:{$\neg a,c$}] (v1) {};
\draw (0+3,-.5) node[dot, blue, label=south:{$\neg a,\neg c$}] (v2) {};
\draw (1+3,-.5) node[dot, green, label=south:{$a$}] (v3) {};

\draw[->] (v)--(v1);
\draw[->] (v)--(v2);
\draw[->] (v)--(v3);

\path (w) edge [dashed,red,bend left=30] node {} (v);
\path (w1) edge [dashed,green, bend right=38] node {} (v3);
\path (w2) edge [dashed,green, bend right=30] node {} (v3);
\path (w3) edge [dashed,blue, bend left=30] node {} (v1);
\path (w3) edge [dashed,blue, bend left=30] node {} (v2);

\end{tikzpicture}
\caption{Witness of joint consistency: dashed lines and colors indicate
  the $\{a\}$-bisimulation.}
\label{fig:witnesses}
	
\end{figure}

Motivated by the lack of the Craig interpolation property, we study the
Craig interpolant existence problem for $\ML$ over $\TT^d$ with $d\geq3$. We show that this problem is harder than $\ML$-separability of $\muML$-formulae over arbitrary models.


\begin{restatable}{thm}{thmmlinterpolants}\label{thm:ml-interpolant-existence}
  For $d\geq 3$, the Craig interpolant existence for $\ML$ over $\TT^d$ is
  \coNExpTime-complete.
  %
\end{restatable}
We start with the upper bound, for any fixed $d\geq 3$. Let 
$\varphi,\varphi'\in\ML^n$ and let $\sigma=\sig(\phi)\cap\sig(\phi')$ be the common signature.
By Lemma~\ref{lem:char-interpolant} and using the observation that witnesses for joint $\bis_\sigma^n$-consistency over $\TT^d$ can be assumed to be of depth at most $n$, the
\emph{non}-existence of a Craig interpolant for $\varphi,\neg\varphi'$ can be decided by a standard ``guess-and-check''-procedure:
\begin{enumerate}

  \item Guess two structures $\M,\M'\in \TT^d$ of depth $n$ in which no point satisfies a non-$\sigma$ proposition;
  
  \item Verify that $\M\models\varphi$, $\M'\models\neg\varphi'$,
    and $\M\bis^n_\sigma M'$.

\end{enumerate}
The runtime of the procedure is exponential: any structure in (1) can be represented in size
polynomial in $|\sigma|$ and $d^n$. Moreover, model checking in modal
logic and bisimulation testing is possible in polynomial time in the
size of the structure and the given formulae. 


The lower bound is more intriguing and relies on an extension of
the counterexample to the Craig interpolation property in the proof of Theorem~\ref{thm:nocraig}. For the sake of simplicity, we do the proof for $d=3$; the proof for $d>3$ is similar. Reconsidering the proof of Theorem~\ref{thm:nocraig} it is
important to note that in \emph{every} witness $\M,\M'$ of joint
$\bis_{\{a\}}^n$-consistency of $\phi,\phi'$, there are two successors
of $\point_I$ that are bisimilar to the same successor of $\point_I'$.
We extend the idea and enforce \emph{exponentially many} bisimilar points by providing suitable 
families of formulae
$(\psi_i)_{i\in \NN},(\psi_i')_{i\in\NN}$.
\begin{lem} \label{lem:aux-psi-lower}
  There are families of formulae $(\psi_i)_{i\in
  \NN},(\psi_i')_{i\in\NN}$ such that:
  \begin{enumerate}

    \item The size of the formulae $\psi_i,\psi_i'$ is polynomial in $i$.

    \item $\sig(\psi_i)=\{a,b_0,\ldots,b_{i-1}\}$ and
      $\sig(\psi_i')=\{a,c\}$.

    \item For every $i\in \NN$, for every $\M,\M'\in
      \TT^3$ with $\M\models \psi_i$, $\M'\models
      \psi_i'$, for every signature $\tau$ with $a\in\tau$ and
      $\{c,b_0,\ldots,b_{i-1}\}\cap\tau=\emptyset$, and every $(\tau,\ell)$-bisimulation $Z$
      witnessing $\M\bis^\ell_{\tau}\M'$ for some $\ell\geq i$, there are points $w_0,\ldots,w_{2^i-1}$
      in depth $i$ in $\M$ and a point $\widehat w$ in depth $i$ in
      $\M'$ such that: 
      \begin{enumerate}
        \item $(w_j,\widehat w)\in Z$ for all $j$,
        \item all $w_j$ and $\widehat w$ satisfy $a$, and
        \item distinct $w_j,w_k$ can be distinguished by some proposition
      in $b_0,\ldots,b_{i-1}$.

      \end{enumerate}
    \end{enumerate}
\end{lem}

\begin{proof}
Define the families $(\psi_i)_{i\in
\NN}$, $(\psi_i')_{i\in \NN}$ of modal formulae inductively as
follows:
\begin{align*}
  \psi_0 & = \psi_0' = a\\ 
  \psi_{i+1} & = \Diamond (a\wedge b_i)\wedge \Diamond (a\wedge \neg
  b_i) \wedge \Box \big( a\to (\psi_i \wedge (b_i\to
  \textstyle\bigwedge_{j<i}\Box^j b_i) \wedge (\neg b_i\to
  \textstyle\bigwedge_{j<i}\Box^j \neg b_i))\big) \\
  %
  %
  \psi_{i+1}' & = \Diamond (\neg a\wedge \textstyle\bigwedge_{j<i}\Box^j\neg a\wedge c)\wedge \Diamond (\neg
  a\wedge \bigwedge_{j<i}\Box^j\neg a\wedge \neg c) \wedge \Diamond (a \wedge \psi_i')
\end{align*}
For Item~(1), it is routine to verify that both $\psi_i$ and $\psi_i'$ are of quadratic size.
Item~(2) is immediate from the definition.  
Item~(3) is shown by induction on $i$, using the same arguments from the proof of Theorem~\ref{thm:nocraig}. The base case $i=0$ is trivial. For the induction step, let $\M,\M'\in \TT^3$ with $\M\models \psi_{i+1}$, $\M'\models \psi_{i+1}'$, let $\tau$ be any signature with $a\in\tau$ and $\{c,b_0,\ldots,b_{i-1}\}\cap\tau=\emptyset$, and let $Z$ be a $(\tau,\ell)$-bisimulation witnessing $\M\bis^\ell_{\tau}\M'$ for some $\ell\geq i+1$. By the definition of $\psi_{i+1}$, there are distinct successors $v_1,v_2$ of the root $\point_I$ of $\M$ both satisfying $\psi_i$. By the definition of $\psi_{i+1}'$, there is a successor $v'$ of the root $\point_I'$ of $\M'$ satisfying $\psi_i'$. By the same argument as in the proof of Theorem~\ref{thm:nocraig}, we must have $v_1 Z v'$ and $v_2 Z v'$. Applying the induction hypothesis to $v_1$ and $v'$ gives us points $w_0,\ldots,w_{2^i-1}$ in depth $i$ below $v_1$ in $\M$ and a point $\widehat w$ in depth $i$ below $v'$ in $\M'$ satisfying 3(a)--(c). 
Applying the induction hypothesis to $v_2$ and $v'$ gives us points $w_0',\ldots,w_{2^i-1}'$ in depth $i$ below $v_2$ in $\M'$ and a point $\widehat w'$ in depth $i$ below $v'$ in $\M'$ satisfying 3(a)--(c). Note that $\widehat w=\widehat w'$ since, by 3(b), both satisfy $a$ and the definition of $\psi_{i+1}'$ ensures that there is a unique element satisfying $a$ at every depth below $v'$. The sequence $w_0,\ldots,w_{2^i-1},w_0',\ldots,w_{2^i-1}'$ together with $\widehat w$ satisfy 3(a)--(c) for $i+1$. This concludes the induction step and thus the proof of Item~(3).
\end{proof}

Intuitively, this means that $\psi_i,\psi_i'$ from Lemma~\ref{lem:aux-psi-lower}
enforce in jointly $\bis_{\{a\}}^i$-consistent models $\M,\M'$ that $\M$
contains $2^i$ points $w_0,\ldots,w_{2^i-1}$ which are all linked to
the same point $\widehat w$ in $\M'$. We exploit this link to
synchronize information between the $w_j$, following a strategy that
has recently been used to show \coNExpTime-hardness for interpolant existence
in some description logics~\cite{DBLP:journals/tocl/ArtaleJMOW23}.

We reduce a suitable \emph{torus tiling problem}. Let $\Delta$ be a (finite) set of
tile types,
$H,V\subseteq \Delta\times \Delta$ be horizontal and vertical compatibility relations, and let $m\in \NN$. Then a function $\tau:\{0,\ldots,m-1\}\times\{0,\ldots,m-1\}\to \Delta$ is a \emph{torus tiling} for $\Delta,H,V,n$ if for all $x,y\in \{0,\ldots,m-1\}$:
\begin{itemize}
\item $(\tau(x,y),\tau(x,y\oplus_m 1))\in H$, and 
\item $(\tau(x,y),\tau(x\oplus_m 1,y))\in V$,
\end{itemize}
where $\oplus_m$ denotes addition modulo $m$. It is well-known that the \emph{exponential torus 
tiling problem}, which takes as input $\Delta,H,V,n$ with $n$ in unary and asks whether there is a torus tiling for $\Delta,H,V,2^n$, is \NExpTime-complete~\cite{DBLP:conf/lam/Furer83}.

We first sketch the proof idea. For a given input $\Delta,H,V,n$, we use the gadget formulae from Lemma~\ref{lem:aux-psi-lower} to enforce the existence of $2^{2n}$ bisimilar points $w_0,\ldots,w_{2^{2n}-1}$. These points shall represent the $2^n\times 2^n$ cells of the exponential torus, which can be suitably addressed using propositions $b_0,\ldots,b_{2n-1}$.
Suppose for the sake of simplicity that the outdegree of the points $w_i$ is
at most $2^{2n}$ (instead of 3). Then we could proceed by enforcing at each $w_i$ with address $(x,y)$ three successors with addresses $(x,y)$, $(x,y\oplus_{2^n} 1)$, and $(x\oplus_{2^n} 1,y)$, and stipulate that they satisfy tile types (encoded using propositions in $\sigma$) compatible with~$H,V$. Then each point of the torus is stipulated three times (from different cells), and due to the bisimilarity of the $w_i$ all stipulated copies of the same point satisfy the same tile type. Hence, existence of a torus tiling will coincide with some joint consistency condition. We will make this now more precise, lifting also the assumption on the outdegree. 

Let $\Delta,H,V,n$ be the input to the exponential torus tiling problem. 
We will provide formulae $\varphi_n,\varphi'_n$ of modal
depth $4n$ and with common signature
$\sigma=\sig(\varphi_n)\cap\sig(\varphi'_n)$ such that:
\[\text{$\varphi_n,\varphi'_n$ are jointly $\bis^{4n}_\sigma$-consistent iff there is a
torus tiling for $\Delta,V,H,2^n$.}\] 
The common signature $\sigma$ 
will consist of propositions $a,b$, and one proposition $t_d$ for
every $d\in \Delta$. Both $\varphi_n$ and $\varphi_n'$ will use
auxiliary propositions to encode counters. The formulae
$\varphi_n,\varphi_n'$ are based on the families of formulae
$(\psi_{i})_{i\in\NN},(\psi_i')_{i\in\NN}$ defined
in Lemma~\ref{lem:aux-psi-lower} and will take the shape:
\begin{align*}
  \varphi_n & = \psi_{2n}\wedge\Box^{2n}\chi_1 \\
  \varphi_n' & = \psi_{2n}'\wedge\Box^{2n}\chi_2
\end{align*}
for formulae $\chi_1,\chi_2$ to be defined below. 

\smallskip Consider models
$\M\models \psi_{2n}$ and $\M'\models \psi_{2n}'$ with
$\M\bis_{\sigma}^{4n}\M'$. Let $w_0,\ldots,w_{2^{2n}-1}$ be the points in
$\M$ and $\widehat w$ be the point in $\M'$ that exist due to
Lemma~\ref{lem:aux-psi-lower}. Recall that, by the lemma, all $w_i$ are linked to
$\widehat w$ by a $(\sigma,2n)$-bisimulation. We
associate two numbers $x_i,y_i$ with each point $w_i$ as follows:

\begin{itemize}

  \item $x_i$ is the number encoded by the valuation of 
    $b_0,\ldots,b_{n-1}$ in $w_i$, and 

  \item $y_i$ is the number encoded by the valuation of 
    $b_n,\ldots,b_{2n-1}$ in $w_i$.

\end{itemize}
Note that by the properties of $w_0,\ldots,w_{2^{2n}-1}$, for every
pair $x,y$ with $0\leq x,y<2^n$, there is
some $w_i$ with $x=x_i$ and $y=y_i$. We denote that point with
$w(x,y)$. Hence, the numbers $x_i,y_i$
can serve as addresses of the $2^n\times 2^n$ cells in the
intended torus tiling. As described above, we will exploit that all $w_i$ are
linked to $w$ by a
$(\sigma,2n)$-bisimulation to synchronize the tile types in each cell. For what follows, it is convenient to denote with
$b_0^i,\ldots,b_{2n-1}^i$ the values of propositions $b_j$ in point
$w_i$, and with $c_0^i,\ldots,c_{2n-1}^i$ and $d_0^i,\ldots,d_{2n-1}^i$ the value of the
propositions $b_j$ in the encoding of $x_i,y_i\oplus_{2^n} 1$ and $x_i\oplus_{2^n} 1,y_i$, respectively. Using auxiliary
non-$\sigma$ propositions, it is not difficult to write formulae
(of polynomial size) in modal logic that express the following Conditions~1
and~2, respectively. Recall that $\sigma$ contains $a,b$ and propositions $t_d$ for
every $d\in \Delta$.
\begin{enumerate}

  \item There are three paths $p_1,p_2,p_3$ of length $2n$ with
    the following properties:

    \begin{enumerate}

      \item each path satisfies $a$ in each point;

      \item on $p_1$, the $j$-th point satisfies $b$ iff
	$b_{j-1}^i=1$, for $1\leq j\leq 2n$;

      \item on $p_2$, the $j$-th point satisfies $b$ iff
	$c_{j-1}^i = 1$, for $1\leq j\leq 2n$;

      \item on $p_3$, the $j$-th point satisfies $b$ iff
	$d_{j-1}^i = 1$, for $1\leq j\leq 2n$;

      \item the ends of $p_1,p_2,p_3$ are labeled with propositions
	$t_{d_1},t_{d_2},t_{d_3}$, respectively, such that $(d_1, d_2)\in H$ and
	$(d_1,d_3)\in V$.

   \end{enumerate}

 \item There is a (ternary) tree of depth $2n$ with the
   following properties: 

   \begin{enumerate}
       
     \item each node has three successors: one satisfying
       $\neg a$, one satisfying $a,b$, and one satisfying $a,\neg b$;

     \item in the leaves of this tree, at most one
       proposition $t_d$ is true. 

   \end{enumerate}

\end{enumerate}
We leave the (rather straightforward) construction of $\chi_2$ to the reader, but provide a few more details on the construction of $\chi_1$. Stipulating three paths of length $2n$ satisfying~(a) is simply done with a modal formula using only $\Diamond$ and $\wedge$, and using non-$\sigma$ propositions to enforce different paths. Conditions~(b)--(d) are realized by propagating the valuation of propositions $b^i_j,c^i_j,d^i_j$ in $w_i$ using non-$\sigma$ propositions to all successors of $w_i$ up to depth $2n$. Finally, Condition~(e) is realized by \enquote{guessing} compatible tiles $d_1,d_2,d_3$ in the root $w_i$ and making them visible only in the end of the paths using $t_{d_1},t_{d_2},t_{d_3}$. Thus, $\chi_1$ takes the form of a disjunction over all compatible choices $d_1,d_2,d_3$. Again, propagation of the guessed tiles to the end of the paths is achieved using non-$\sigma$ propositions.

This finishes the definition of the formulae $\varphi_n,\varphi_n'$
and we can proceed with showing the correctness of the reduction.

\begin{clm} 
 The formulas $\psi_{2n}\wedge \Box^{2n}\chi_1$
und $\psi'_{2n}\wedge\Box^{2n}\chi_2$ are jointly
$\bis_\sigma^{4n}$-consistent
iff there is a torus tiling for $\Delta,H,V,2^n$.
\end{clm}

\smallskip To prove the claim, suppose first that $\psi_{2n}\wedge
\Box^{2n}\chi_1$ und $\psi'_{2n}\wedge\Box^{2n}\chi_2$ are jointly
$\bis_\sigma^{4n}$-consistent, witnessed by models $\M,\M'$. Let also $w_0,\ldots,w_{2^{2n}-1}$ and $\widehat w$ be the points that exist in $\M,\M'$
due to Lemma~\ref{lem:aux-psi-lower}.

We define a torus tiling $\tau$ as follows. Let $x,y$ be any cell of the exponential torus, that
is, $0\leq x,y<2^n$.  Conditions~1b--1d enforce that three paths are
stipulated: one (via~1b) in point $w(x,y)$, one (via~1c) in point
$w(x\ominus_{2^n} 1,y)$, and one (via~1d) in point $w(x,y\ominus_{2^n}1)$, where,
similar to $\oplus_{2^n}$, $\ominus_{2^n}$ denotes subtraction modulo $2^n$. Each
of these paths is labeled, using $b\in\sigma$ with (the encoding of) $x,y$ along its
elements.  Due to Conditions~1a and~2a, these three paths (in $\M$) can only be
bisimilar to the path (in $\M'$) in the tree stipulated below
$\widehat w$ that is labeled
with (the encoding of) $x,y$. In particular, the ends of the paths are
bisimilar to the same leaf in the tree. Since, by
Condition~2b, every leaf in the tree satisfies at most one $t_d$, and,
by Condition~1e, the end of each path satisfies at least one
$t_d$, all ends are labeled with the same $t_d$. We set $\tau(x,y)=d$. Synchronization is
then achieved by Condition~1e. Indeed, consider $(x,y)$ and
$(x\oplus_{2^n} 1,y)$. Then Condition~1e ensures that the end of the paths for
$(x,y)$ and $(x\oplus_{2^n} 1,y)$ stipulated in $w(x,y)$ are labeled with $t_d,t_{d'}$ such that
$(d,d')\in V$. The argument for horizontal compatibility is symmetric.

In the other direction, it is not difficult to construct jointly
$\bis^{4n}_\sigma$-consistent models from a given torus tiling for
$\Delta,H,V,2^n$. This concludes the proof of the claim and thus of Theorem~\ref{thm:ml-interpolant-existence} as well.

We show next that the situation for the more general $\ML$-separability, for which the input formulae are in $\muML$, is even worse. 
 
\begin{restatable}{thm}{thmternarycomplexity}\label{thm:ternary complexity}
  For every $d\geq 3$, $\ML$-separability of $\muML$-formulae over $\TT^d$ is
  $\TwoExpTime$-complete.
  %
\end{restatable}

Thus, over $\TT^d$ for $d\geq 3$, $\ML$-separability is provably 
harder than $\ML$-definability---recall that $\ML$-definability was shown to be \ExpTime-complete in Theorem~\ref{thm:otto main}. 
Both the lower and the upper bound of Theorem~\ref{thm:ternary
complexity} are non-trivial. We provide proofs for both of them in the following two subsections: first the lower bound in Subsection~\ref{sec:lower} and then the upper bound in Subsection~\ref{sec:upper}. We also observe in the end of the latter Subsection~\ref{sec:upper} that the procedure can be easily made constructive.

\subsection{Proof of Lower Bound in Theorem~\ref{thm:ternary complexity}}\label{sec:lower}

We proceed in two steps. We first show a lower bound for a slightly modified separability problem, and then reduce this to $\ML$-separability. In detail, \emph{signature-restricted $\ML$-separability of $\muML$-formulae} is the problem that takes as an input $\muML$-formulae $\varphi,\varphi'$ and a signature $\sigma$ and decides whether there is an $\ML_\sigma$-separator for $\varphi,\varphi'$. Hence, it is the variant of $\ML$-separability in which we can additionally restrict the signature of the sought separator.

\begin{lem} \label{lem:craigmodalsep} 
  For any $d\geq 3$, signature-restricted $\ML$-separability of $\muML$-formulae over $\TT^d$ is
  \TwoExpTime-hard.
\end{lem}

Again, we show the lower bound only for $d=3$; the proof for $d>3$ is analogous. 
The overall proof strategy is inspired by recently studied interpolant existence problems for description logics~\cite{DBLP:journals/tocl/ArtaleJMOW23}, but the details have to be adapted to the present setup. The proof is via a reduction of the word problem for
languages recognized by exponentially space bounded, alternating
Turing machines, which we recall next. 

An \emph{alternating Turing machine (ATM)} is a tuple $\ATM=(Q,\Theta,\Gamma,q_0,\Delta)$ where
$Q=Q_{\exists}\uplus Q_{\forall}$
is a finite
set of states partitioned into \emph{existential
states}~$Q_{\exists}$ and \emph{universal states}~$Q_{\forall}$.
Further,~$\Theta$ is the input alphabet and $\Gamma$ is the tape
alphabet that contains a \emph{blank symbol} $\textsf{blank} \notin \Theta$,
$q_0\in Q_{\forall}$ is the \emph{initial state}, and $\Delta\subseteq
Q\times \Gamma\times Q\times \Gamma \times \{L,R\}$ is the
\emph{transition relation}.  We assume without loss of generality that
the set $\Delta(q,a):=\{(q',a',M)\mid (q,a,q',a',M)\in\Delta\}$
contains exactly zero or two transitions for every $q\in Q_\exists\cup Q_\forall$ and $a \in
\Gamma$. 
Moreover, the state $q'$ in the transition $(q,a,q',a',M)$ must be from $Q_\forall$ if $q \in
Q_\exists$ and from $Q_\exists$ otherwise, that is, existential and
universal states alternate. A
\emph{configuration} of an ATM is a word $wqw'$ with
\mbox{$w,w'\in\Gamma^*$} and $q\in Q$. 
We say that $wqw'$ is \emph{existential} if~$q$ is, and likewise for
\emph{universal}.  \emph{Successor configurations} are defined in the
usual way.  Note that every configuration has exactly zero or two
successor configurations. 

We define next when an ATM accepts an input. Note that according to our definition an ATM does not have accepting or rejecting states. Instead, acceptance is defined in terms of the existence of infinite computations, which is slightly non-standard but more convenient for our purposes (and easily seen to be equivalent to the standard definition). A \emph{computation tree} of an ATM \ATM on input $w$ is a (possibly infinite) tree whose nodes are labeled with
configurations of \ATM such that
\begin{itemize}

  \item the root is labeled with the initial configuration
    $q_0w$;

  \item if a node is labeled with an existential configuration 
    $wqw'$, then it has a single successor which is labeled 
    with a successor configuration of~$wqw'$;

  \item if a node is labeled with a universal configuration $wqw'$,
    then it has two successors which are labeled with the two
    successor configurations of~$wqw'$.

\end{itemize}
An ATM \ATM \emph{accepts} an input $w$ if there is a
computation tree of \ATM on $w$ in which every branch is
infinite. 
An ATM is called \emph{$f$-space bounded} for some function $f:\NN\to\NN$ if in every configuration of every computation tree of \ATM on an input of length $n$ the non-blank part of the tape is of
length at most $f(n)$. The class $\AExpSpace$ is the class of all languages recognizable by some $2^{p(n)}$-space bounded ATM, for a polynomial $p$. It is well-known that $\AExpSpace=\TwoExpTime$, which means that it suffices to reduce the word problem for languages recognized by $2^{p(n)}$-space bounded ATMs~\cite{DBLP:journals/jacm/ChandraKS81}.

For the reduction, let $\ATM=(Q,\Theta,\Gamma,q_0,\Delta)$ be any $2^{p(n)}$-space bounded ATM, for some fixed polynomial $p$. 
Let further $w = a_0\ldots a_{n_0-1}$ be an input of length $n_0$, and set $n=p(n_0)$. We will provide formulae $\varphi_n,\varphi'_n$ and signature $\sigma$ such that 
\[\text{$\varphi_n,\varphi'_n$ are jointly $\bis^m_\sigma$-consistent
for every $m\in \NN$ \quad iff\quad \ATM accepts $w$.}\] 
This suffices by Equivalence~\eqref{eq:non-sep vs n-bis}. Formulae $\varphi_n,\varphi'_n$ use the family of formulae $\psi_i,\psi_i'$ defined in
Lemma~\ref{lem:aux-psi-lower}, and are defined as
\begin{align*}
  \varphi_n & = \psi_n \wedge \Box^n \chi \\
  \varphi_n' & = \psi_n'
\end{align*}
for a formula $\chi$ to be defined below. The signature
$\sigma$ is given by
\[\sigma=\{a,z,p_\exists^1,p_\exists^2,p_\forall\}\cup \{c_\alpha\mid \alpha \in \Gamma\cup (Q\times\Gamma)\}.\]
Here, proposition $a$ is the one used in the gadget formulae $\psi_n,\psi_n'$, $z$ is an auxiliary
proposition whose meaning becomes clear below, propositions $p_\exists^*, p_\forall$ are responsible for marking configurations existential and universal, respectively, and finally each $c_\alpha$ describes a possible content of a single cell in a configuration: it is just a letter or a letter together with a state, also indicating the head position of the configuration. 

To explain the main idea behind formula $\chi$ in $\varphi_n$, consider models $\M,\M'$ witnessing joint $\bis_\sigma^m$ consistency
of $\varphi_n,\varphi_n'$ for some sufficiently large $m\geq n$. By
Lemma~\ref{lem:aux-psi-lower}, there are points $w_0,\ldots,w_{2^n-1}$ in depth of $n$ $\M$
and $\widehat w$ in $\M'$ which are linked by a $(\sigma,m-n)$-bisimulation. Additionally, they all satisfy $\chi$, due to the shape of $\varphi_n$. The job of $\chi$ is to enforce below each $w_i$ a computation tree of \ATM on input $w$, as illustrated in Figure~\ref{fig:computation-tree}. Using auxiliary non-$\sigma$-propositions, it is routine to give a $\muML$-formula which:
\begin{itemize}

  \item enforces the skeleton of a computation tree for $\ATM$, in
    which each configuration is modeled by a path of length $2^n$,
    
  \item enforces that each point of the skeleton satisfies exactly one of the $\sigma$-propositions $c_\alpha$, and

   \item enforces that universal configurations have two successors and that universal and existential configurations alternate. 

\end{itemize}
At this point, it is a good moment to mention that the propositions $p_\exists^1,p_\exists^2$ not only mark that the current configuration is existential, but also that they are the first or respectively second successor of a universal configuration, see beginning of the two existential configurations in Figure~\ref{fig:computation-tree}.

The challenge is to enforce that all successor configurations are consistent with $\ATM$'s transition relation $\Delta$, and we exploit here the bisimilarity of the $w_i$ to achieve this. The idea is that, below each $w_i$ we coordinate only the $(2^n-i)$-th position of each configuration with its successor configuration(s) in the computation tree, as depicted in Figure~\ref{fig:computation-tree}.\footnote{The figure contains some additional labeling that will be explained in more detail later.} This is possible using auxiliary non-$\sigma$-propositions and due to the fact that we can access the index $i$ of $w_i$ via propositions $b_0,\ldots,b_{n-1}$ mentioned in Lemma~\ref{lem:aux-psi-lower}. Note that this provides sufficient coordination, since, due to the bisimilarity of the $w_i$, \emph{all} the positions will be synchronized. As mentioned above, this overall strategy is similar to one employed in~\cite{DBLP:journals/tocl/ArtaleJMOW23}.

%
%
%
%
%
%
%
\begin{figure}[t]
  \centering
	\begin{tikzpicture}
		\tikzset{
			dot/.style = {draw, fill=black, circle, inner
			sep=0pt, outer sep=1pt, minimum size=3pt},
			wdot/.style = {draw, fill=white, circle, inner
			sep=0pt, outer sep=1pt, minimum size=3pt}
		}
		

\draw (-1,0.75) node[label=$T_{i}$] (Ti) {};
		
\draw (1.2,0.75) node[label={$\text{universal conf.}$}] (Ti) {};
\draw (5,-.3) node[label={$\text{existential conf.}$}] (Ti) {};
\draw (8.2,-.3) node[label={$\text{universal conf.}$}] (Ti) {};


\draw (0,0) node[wdot, label=south:$p_\forall$, label=west:$w_i$] (a2) {};
\draw (0.75,0) node[dot, label=$$] (a3) {};
\draw (1.125,-0.275) node[label=$\ldots$] (dot1) {};
\draw (1.5,0) node[dot, label=south:$\!\!\!\!\!\!2^n -i$] (f) {};
\draw (1.875,-0.275) node[label=$\ldots$] (dot2) {};
\draw (2.25,0) node[dot, label=$$] (a4) {};
\draw (3,0) node[dot, label=$$] (a5) {};

\draw (3.75,-1) node[wdot, label=north:$p_\exists^{2}$] (aa1) {};
\draw (4.5,-1) node[dot, label=$$] (ii1) {};
\draw (4.875,-1.25) node[label=$\ldots$] (dota1) {};
\draw (5.25,-1) node[dot, label=south:$2^n-i$] (aa2) {};
\draw (5.625,-1.25) node[label=$\ldots$] (dota2) {};
\draw (6,-1) node[dot,label=$$] (aa3) {};
\draw (6.75,-1) node[wdot,label=north:$p_\forall$] (aa4) {};
\draw (7.5,-1) node[dot, label=$$] (kk1) {};
\draw (7.875,-1.25) node[label=$\ldots$] (dota3) {};
\draw (8.25,-1) node[dot, label=south:$2^n-i$] (abk1) {};
\draw (8.625,-1.25) node[label=$\ldots$] (dota4) {};
\draw (9,-1) node[dot, label=$$] (aa5) {};
\draw (9.75,-1) node[dot, label=$$] (aa6) {};
\draw (10.5,-1.5) node[wdot, label=$$] (aaa1) {};
\draw (10.875,-1.775) node[label=$\ldots$] (dotaaf) {};
\draw (10.5,-0.5) node[wdot, label=$$] (aab1) {};
\draw (10.875,-0.75) node[label=$\ldots$] (dotabf) {};

\draw (3.75,1) node[wdot, label=south:$p_\exists^{1}$] (ab1) {};
\draw (4.5,1) node[dot, label=$$] (ii2) {};
\draw (4.875,0.725) node[label=$\ldots$] (dotb1) {};
\draw (5.25,1) node[dot, label=south:$2^n-i$] (ab2) {};
\draw (5.625,0.725) node[label=$\ldots$] (dotb2) {};
\draw (6,1) node[dot,label=$$] (ab3) {};
\draw (6.75,1) node[wdot,label=south:$p_\forall$] (ab4) {};
\draw (7.5,1) node[dot, label=$$] (kk2) {};
\draw (7.875,0.725) node[label=$\ldots$] (dotb3) {};
\draw (8.25,1) node[dot,label=south:$2^n-i$] (abk2) {};
\draw (8.625,0.725) node[label=$\ldots$] (dotb4) {};
\draw (9,1) node[dot, label=$$] (ab5) {};
\draw (9.75,1) node[dot, label=$$] (ab6) {};
\draw (10.5,0.5) node[wdot, label=$$] (aba1) {};
\draw (10.875,0.225) node[label=$\ldots$] (dotbaf) {};
\draw (10.5,1.5) node[wdot, label=$$] (abb1) {};
\draw (10.875,1.25) node[label=$\ldots$] (dotbbf) {};


\draw[->, >=stealth] (a2) edge (a3);
\draw[->, >=stealth] (a4) edge (a5) node[label={[shift={(0.35,0.2)}]:$$}] {};
\draw[->, >=stealth] (a5) edge (aa1) node[label={[shift={(0.35,0.2)}]:$$}] {};
\draw[->, >=stealth] (a5) edge (ab1) node[label={[shift={(0.35,0.2)}]:$$}] {};
\draw[->, >=stealth] (aa1) edge (ii1) node[label={[shift={(0.35,0.2)}]:$$}] {};
%
\draw[->, >=stealth] (aa3) edge (aa4) node[label={[shift={(0.35,0.2)}]:$$}] {};
\draw[->, >=stealth] (aa4) edge (kk1) node[label={[shift={(0.35,0.2)}]:$$}] {};
\draw[->, >=stealth] (aa5) edge (aa6) node[label={[shift={(0.35,0.2)}]:$$}] {};
\draw[->, >=stealth] (aa6) edge (aaa1) node[label={[shift={(0.35,0.2)}]:$$}] {};
\draw[->, >=stealth] (aa6) edge (aab1) node[label={[shift={(0.35,0.2)}]:$$}] {};

\draw[->, >=stealth] (ab1) edge (ii2) node[label={[shift={(0.35,0.2)}]:$$}] {};
%
\draw[->, >=stealth] (ab3) edge (ab4) node[label={[shift={(0.35,0.2)}]:$$}] {};
\draw[->, >=stealth] (ab4) edge (kk2) node[label={[shift={(0.35,0.2)}]:$$}] {};
\draw[->, >=stealth] (ab5) edge (ab6) node[label={[shift={(0.35,0.2)}]:$$}] {};
\draw[->, >=stealth] (ab6) edge (aba1) node[label={[shift={(0.35,0.2)}]:$$}] {};
\draw[->, >=stealth] (ab6) edge (abb1) node[label={[shift={(0.35,0.2)}]:$$}] {};

			
\path[black, dashed, bend right] (f) edge (aa2);
\draw (2.5,0.15) node[label=$$] () {};
\path[black, dashed, bend right] (aa2) edge (abk1);
\draw (2.5,0.15) node[label=$$] () {};

\path[black, dashed, bend left] (f) edge (ab2);
\draw (2.5,0.15) node[label=$$] () {};
\path[black, dashed, bend left] (ab2) edge (abk2);
\draw (2.5,0.15) node[label=$$] () {};

\end{tikzpicture}	
		
\caption{Model of a computation tree of $\ATM$, enforced by formula $\chi$.}
\label{fig:computation-tree}
	
\end{figure}



We will now provide $\chi$ more concretely. It is a conjunction $\chi
= \chi_0\wedge \Box^*\chi_1$, where $\Box^*\chi_1$ is an abbreviation for the $\muML$-formula $\nu x. (\chi_1\wedge \Box x)$ expressing that $\chi_1$ holds in the entire subtree below the current point. Among others, $\chi$ uses propositions $c_0,\ldots,c_{n-1}$ to implement an exponential counter called the \emph{C-Counter} to address the cells in all configurations, and propositions $p_\forall,p_{\exists}^1,p_{\exists}^2$ to mark universal and existential configurations. Formula $\chi_0$ initializes the
C-counter to $0$ and marks the first configuration as universal (recall that $q_0\in
Q_\forall$) using proposition $p_\forall$ as follows:
\[\chi_0 = \neg c_0\wedge\ldots\wedge \neg c_{n-1}\wedge p_\forall.\] 
Formula $\chi_1$ is a conjunction of several formulae in turn. The more standard ones are detailed in the following list. We use $(C=i)$, or similar expressions, as an abbreviation for the
Boolean combination of the propositions $c_0,\ldots,c_{n-1}$ that encodes value
$i$.
\begin{itemize}

  \item One conjunct is responsible for incrementing the C-counter modulo $2^n$ along edges in models. This is entirely standard, so we refrain from detailing it.

  \item The conjunct $\Diamond\top$ enforces infinite trees.

  \item The following conjuncts enforce the structure of a computation tree of \ATM:
\begin{align}
  (C<2^n-1) \wedge p_\forall & \to \Box p_{\forall} \label{eq:startone}\\
  (C<2^n-1) \wedge p_\exists^i & \to \Box p_{\exists}^i && i\in\{1,2\} \\
  (C=2^n-1) \wedge p_\forall & \to \Box(p_{\exists}^1\vee p_{\exists}^2) \\
  (C=2^n-1) \wedge p_\exists^i& \to \Box p_\forall && i\in\{1,2\} \\
  (C=2^n-1) \wedge p_{\forall} & \to \Diamond z\wedge \Diamond \neg z \label{eq:endone}
\end{align}
These implications enforce that all points which represent a configuration
satisfy one of $p_{\forall},p_{\exists}^1,p_{\exists}^2$
indicating the kind of configuration and, if existential, also a choice of the transition
function. The symbol $z\in \Sigma$ enforces the branching.
We refer the reader once again to Figure~\ref{fig:computation-tree} for an illustration of the enforced structure. Note that $\circ$ marks the beginning of a configuration, that is, where the C-counter is $0$.

\item The initial configuration on input $w=a_0\ldots a_{n-1}$ is enforced
by 
\[c_{q_0,a_0}\wedge \Diamond (c_{a_1}\wedge \Diamond (c_{a_2} \wedge (
\ldots \wedge \Diamond (c_{a_{n-1}}\wedge \Diamond
\chi_\textit{blank})\ldots))),\]
where $\chi_{\textit{blank}}$ enforces label $c_{\textsf{blank}}$
until the end of the configuration. Again, this is entirely standard so we omit the details.

\item One conjunct enforces that each point is labeled with exactly one of the propositions $c_\alpha$, $\alpha\in \Gamma\cup (Q\times \Gamma)$.

\end{itemize}

The remaining conjuncts of $\chi_1$ are responsible for the
coordination of the successor configurations. Recall that the goal is to coordinate in the subtree below $w_i$ the $2^n-i$-th position of each configuration with the
$2^n-i$-th position of its successor configuration(s). As announced, we exploit the propositions $b_0,\ldots,b_{n-1}$ from Lemma~\ref{lem:aux-psi-lower} for this purpose, assuming that in each $w_i$ the value encoded by these propositions is exactly $i$. We refer to the counter encoded by propositions $b_0,\ldots,b_{n-1}$ as the \emph{B-counter}. Analogously to the C-counter, the B-counter is incremented modulo $2^n$ along edges in models (enforced using another conjunct in $\chi_1$).



To coordinate successor configurations, we associate with $\ATM$
functions $f_i$, $i\in \{1,2\}$ that map the content of three
consecutive cells of a configuration to the content of the middle cell
in the $i$-th successor configuration (assuming an arbitrary order on
the set $\Delta(q,a)$, for all $q,a$). In what follows, we ignore the corner cases that occur
at the border of configurations; they can be treated in a similar way. 
Clearly, for each possible triple
$(\alpha_1,\alpha_2,\alpha_3)\in (\Gamma\cup(Q\times \Gamma))^3$, the
$\ML$-formula $\varphi_{\alpha_1,\alpha_2,\alpha_3}=c_{\alpha_1}\wedge
\Diamond (c_{\alpha_2}\wedge \Diamond c_{\alpha_3})$ is true at a
point $v$ of the computation tree iff $v$ is labeled with
$c_{\alpha_1}$, a successor $u$ of $v$ is labeled with
$c_{\alpha_2}$, and a successor $t$ of $u$ is labeled with
$c_{\alpha_3}$. In each configuration, we synchronize points
with B-counter $0$ by including for every $(\alpha_1,\alpha_2,\alpha_3)$ and $i\in\{1,2\}$ the following
implications: 
\begin{align}
  (B=2^n-1) \wedge (C<2^n-2) \wedge \varphi_{\alpha_1,\alpha_2,\alpha_3}
  \wedge p_{\forall} & \to
  \Box q^1_{f_1(\alpha_1,\alpha_2,\alpha_3)}\wedge \Box
  q^2_{f_2(\alpha_1,\alpha_2,\alpha_3)}\label{eq:starttwo}\\
  (B=2^n-1) \wedge (C<2^n-2) \wedge
  \varphi_{\alpha_1,\alpha_2,\alpha_3} \wedge  p_\exists^i
  & \to \Box q^i_{f_i(\alpha_1,\alpha_2,\alpha_3)}\label{eq:endtwo}
\end{align}
At this point, the importance of the superscript in
$p_\exists^\ast$ becomes apparent: since different cells of a
configuration are synchronized
in trees below different $w_k$ the superscript makes sure that all trees
rely on the same choice for existential configurations. Propositions
$q^i_{\alpha}$ are used as markers (not in $\sigma$) and
are propagated for $2^n$ steps, exploiting the C-counter.
The superscript $i\in\{1,2\}$ determines the successor configuration
that the symbol is referring to. After crossing the end of a
configuration, the symbol $\alpha$ is propagated using propositions
$q_{\alpha}'$ (the superscript is not needed anymore because the
branching happens at the end of the configuration, based on $z$). This is achieved by adding the following implications as conjuncts of $\chi_1$:
\begin{align}
  (C<2^n-1) \wedge q_\alpha^i & \to \Box q_{\alpha}^i \label{eq:startthree}\\
  (C=2^n-1) \wedge p_\forall \wedge q_\alpha^1 & \to \Box(z\to
  q'_\alpha)\\
  (C=2^n-1) \wedge p_\forall \wedge q_\alpha^2 & \to \Box
  (\neg z\to q'_\alpha)\\
  (C=2^n-1) \wedge p_\exists^i \wedge q_\alpha^i & \to \Box
  q'_\alpha && i\in\{1,2\}\\
  (B<2^n-1) \wedge q'_\alpha & \to \Box q'_\alpha \\
  (B=2^n-1) \wedge q'_\alpha & \to \Box q_\alpha\label{eq:endthree}
\end{align}
 
For those $(q,a)$ with $\Delta(q,a)=\emptyset$, we add the conjunct
\[\neg c_{q,a}.\]

It remains to establish correctness of the reduction.

\begin{clm}\label{claim:sim1}
  The following conditions are equivalent:
  \begin{enumerate} 

    \item $\ATM$ accepts $w$;

    \item $\varphi_n,\varphi_n'$ are jointly $\bis_\sigma^m$-consistent
      for every $m\in\NN$.

  \end{enumerate} 		

\end{clm}

\begin{proof} ``1 $\Rightarrow$ 2''. If $M$ accepts $w$, there is a
  computation tree of $\ATM$ on $w$. We construct two models
  $\M\models\varphi_n$ and $\M'\models \varphi_n'$ such that $\M\bis_\sigma
  \M'$, which implies joint $\bis_\sigma^m$-consistency for every $m$.
  Let $\widehat\M$ be the infinite tree-shaped
  model that represents the computation tree of $\ATM$ on $w$ as
  described above, that is, configurations are represented by
  sequences of $2^n$ points and labeled by
  $p_\forall,p_\exists^1,p_\exists^2$ depending on whether the
  configuration is universal or existential, and in the latter case
  the superscript indicates which choice has been made for the
  existential state. Finally, the first point of the first successor
  configuration of a universal configuration is labeled with
  $z$. Observe that $\widehat\M$
  interprets only the symbols in $\sigma$ as non-empty. Now, we
  obtain models $\M_k$, $k<2^n$ from $\widehat\M$ by interpreting
  non-$\sigma$-symbols as follows: 
  \begin{itemize}

    \item the C-counter starts at $0$ at the root and counts modulo
      $2^n$ along each path starting in the root;

    \item the B-counter starts at $k$ at the root and counts modulo
      $2^n$ along each path starting in the root;

    \item the auxiliary propositions of the shape $q_\alpha^i$ and
      $q_\alpha'$ are interpreted in a minimal way as to satisfy
      implications~\eqref{eq:startone}--\eqref{eq:endthree}. Note that, since these implications are essentially Horn formulae, there is a unique result.

  \end{itemize}
  Now, $\M'$ is defined as follows:

  \begin{itemize}

    \item start with a path of length $n$ in which each node satisfies
      $a$,

    \item add one successor satisfying $\neg a,\neg c$ and one successor
      satisfying $\neg a,c$ to each node in the path, and

    \item at the end $\widehat w$ of the path, attach a copy of $\widehat M$.

  \end{itemize}
  Next, obtain $\M$ from the $\M_k$ as follows:
  \begin{itemize}

    \item Start with a full binary tree of depth $n$,
    
    \item add one successor satisfying $\neg a$ to each node in the
      tree, 

    \item interpret propositions $b_0,\ldots,b_{n-1}$
      in a way such that the B-counter values of the $2^n$ leaves
      $w_0,\ldots,w_{2^n-1}$ (of the original binary tree) range from
      $0$ to $2^{n}-1$, and

    \item attach at each leaf $w_k$ the tree $\M_k$.

  \end{itemize}
  It can be verified that the reflexive, transitive, and symmetric
  closure of 
  \begin{itemize}

    \item all pairs $(u,v)$ for points $u$ satisfying $a$ in $\M$ at level $i$
      and points $v$ satisfying $a$ in $\M'$ at level $i$,

    \item all pairs $(u,v)$ for points $u$ satisfying $\neg a$ in $\M$ at level $i$
      and points $v$ satisfying $\neg a$ in $\M'$ at level $i$, and

    \item all pairs $(v,v')$, with $v$ in $\widehat\M$ and $v'$ a
      copy of $v$ in some tree $\M_{k}$

  \end{itemize}
  witnesses $\M\bis_\sigma\M'$.

  ``2 $\Rightarrow$ 1''. Suppose $\varphi_n,\varphi_n'$ are jointly
  $\bis_\sigma^m$-consistent for every $m\in \NN$. We aim to construct models $\M,\M'$
  with $\M\models\varphi$, $\M'\models\varphi'$, and $\M\bis_\sigma\M'$. Let us
  assume for a moment we have such $\M,\M'$. Since $\M\models\psi_n$ and
  $\M'\models\psi_n'$, by Lemma~\ref{lem:aux-psi-lower}, there are pairwise
  $\sigma$-bisimilar points $w_0,\ldots,w_{2^n-1}$ in depth $n$. That is, the
  trees starting at $w_0,\ldots,w_{2^n-1}$ are bisimilar. By the shape of
  $\varphi,\varphi'$, these trees are additionally models of
  $\chi_0\wedge\Box^*\chi_1$. It follows that in the tree below $w_k$, the cell
  contents of the $(2^n-k)$-th cell is coordinated, between any two consecutive
  configurations. Overall, all cell contents are coordinated and thus all trees
  below some $w_k$ contain a computation tree of $\ATM$ on input $w$ (which is
  solely represented with $\sigma$-symbols). Thus $\ATM$ accepts $w$.
  
  It remains to construct $\M,\M'$. We use the standard technique of \enquote{skipping
  bisimulations} exploiting that we work over models of finite outdegree. Let
  $\M_m,\M_m'$ be witnesses for joint $\bis_\sigma^m$-consistency of
  $\varphi_n,\varphi_n'$, for every $m\in\mathbb{N}$, and let $Z_m$ be a witnessing bisimulation. We can assume without loss of generality that all $\M_m,\M_m'$ are full ternary trees and that propositions outside $\sig(\varphi_n)\cup\sig(\varphi_n')$ are globally false. Hence, there are only finitely many possible valuations of points in $\M_m,\M_m'$. 

  We will construct a sequence of models $\N_m,\N_m'$, $m\in\mathbb{N}$ and $(\sigma,m)$-bisimulations $Z_m'$ such that for all $m\in \mathbb{N}$ and all $i\leq m$:
  \begin{enumerate}
    \item[($\ast$)] $\N_i$ and $N_m$, $\N_i'$ and $\N_m'$, and $Z_i$ and $Z_m$ coincide when restricted to the first $i$ levels. 
  \end{enumerate}
  We proceed by induction. Since the roots of the $\M_m,\M'_m$ may only have one of finitely many valuations, there is a pair of valuations that occurs infinitely often in the roots. We \enquote{skip} all those $\M_m,\M_m'$ which do not have the fixed valuations in their roots, ending up with a sequence that satisfies ($\ast$) for all $m\in\mathbb{N}$ and $i=0$.
  
  In the inductive step, suppose we have constructed a sequence that satisfies ($\ast$) for every $m$ and all $i\leq i_0$ for some $i_0$. Hence, for every $j\geq i_0$, all points in level $i_0$ of $\N_j,\N_j'$ have the same valuations as in $\N_{i_0},\N_{i_0}'$. Moreover, $uZ_jv$ iff $uZ_mv$ for all points $u,v$ in level $i_0$, and $j\geq i_0$. Since there are only finitely many points in level $i$ and each point has bounded outdegree, there are only finitely many possibilities of choosing the valuations of and bisimulations among the (finitely many) points in level $i_0+1$. Hence, one such choice is realized infinitely often, and we can skip all $\N_m,\N_m'$ with $m>i_0$ that do not realize this choice. After this modification, the sequence satisfies ($\ast$) for all $m$ and all $i\leq i_0+1$.

  Let $\widehat\N_{m},\widehat\N_{m}'$ denote the $m$-prefixes of $\N_m,\N_m'$, respectively. (Using our notation, this means $\widehat\N_m=(\N_m)_{|m}$ and $\widehat\N'_m=(\N'_m)_{|m}$.) Let further be $\widehat Z_m$ the restriction of $Z_m$ to these $m$-prefixes.
  It is tedious but not difficult to verify that
  \[\M=\bigcup_{m\in\mathbb{N}} {\widehat\N_{m}}, \quad \M'=\bigcup_{m\in\mathbb{N}} {\widehat\N_{m}'},\quad\text{and}\quad Z=\bigcup_{m\in\mathbb{N}}\widehat Z_{m}\]
  satisfy $\M\bis_\sigma\M'$, $\M\models\varphi_n$, and $\M'\models\varphi_n'$. The latter rely on the specific shape of $\varphi_n,\varphi'_n$ which contain only a largest fixpoint (in form of the $\Box^*$).
\end{proof}

It remains to show that signature-restricted $\ML$-separability reduces to $\ML$-separability; we show it only over classes $\TT^d$, for $d\geq 2$, but the reduction can easily be modified to work also over the class of all models. 
\begin{lem}\label{lem:craig-to-general}
  For any $d\geq 2$, there is a polynomial time reduction of 
  signature-restricted $\ML$-separability of $\muML$-formulae over $\TT^d$ to $\ML$-separability of $\muML$-formulae over $\TT^d$.
\end{lem}

\begin{proof}
  Let $d\geq 2$, $\varphi,\varphi'\in \muML$ and $\sigma$ be a
  signature. We construct $\muML$-formulae
  $\widehat \varphi,\widehat\varphi'$ such that
  \begin{equation}
    \text{$\varphi,\varphi'$ are
      $\ML_\sigma$-separable over $\TT^d$ \quad iff\quad $\widehat\varphi,\widehat\varphi'$ are
    $\ML$-separable over $\TT^d$.}
    \label{eq:correct-craig-to-general}
  \end{equation}
  Let $\sigma_\varphi=\sig(\varphi)\setminus\sigma$ and
  $\sigma_{\varphi'}=\sig(\varphi')\setminus\sigma$,
  respectively, be the sets of all
  propositions that occur in $\varphi$, respectively $\varphi'$, but not in $\sigma$. Let $\widehat a$ be a proposition that does not occur in $\varphi,\varphi'$. Then, set 
  \begin{align*}
    \widehat \varphi & = h(\varphi,\sigma_\varphi) \\
    \widehat \varphi' & = h(\varphi',\sigma_{\varphi'})
  \end{align*}
  where $h(\psi,\tau)$ is the formula obtained from $\psi$ by 
  \begin{itemize}

    \item first inductively replacing each subformula $\Diamond \theta$ with
      $\Diamond \widehat a\wedge \Box (\widehat a\to\Diamond(\widehat a\wedge \theta))$ and each subformula
      $\Box\theta$ with $\Box (\widehat a\to \Box (\widehat a\to \theta))$, and

    \item then replacing each proposition $c\in
      \tau$ with $\Diamond (\neg \widehat a\wedge \Diamond^* c)$, where $\Diamond^*c$ is the $\muML$ formula that expresses that a point satisfying $c$ is reachable.

  \end{itemize}
  Intuitively, in the first step $h$ relativizes $\psi$ to $\widehat a$ and \enquote{skips} every
  second level, and the second step replaces non-$\sigma$ propositions
  by a formula that is not modally definable and not visible in the
  $\widehat a$-relativization. Clearly, the size of $h(\psi,\tau)$ is
  polynomial in the size of its inputs. We show correctness of the construction as stated in Equivalence~\eqref{eq:correct-craig-to-general}.

  \smallskip
  For ``$\Leftarrow$'', suppose that $\varphi,\varphi'$ are not
  $\ML_\sigma$-separable. By Equivalence~\eqref{eq:non-sep vs n-bis}, for every
  $n\in\NN$, there are models
  $\M_n\models\varphi,\M_n'\models\varphi'$ with
  $\M_n\bis^n_\sigma\M_n'$. 
  We construct models
  $\N_n\models\widehat\varphi,\N_n'\models \widehat\varphi'$ with
  $\N_n\bis^{2n}\N_n'$ as follows. The model $\N_n=(N_n,\point_I,\to_N,\val_N)$ is obtained from
  $\M_n=(M_n,\point_I,\to_M,\val_M)$ as follows:  
  \begin{itemize}

    \item it has universe \[N_n=\{u,u',u_i\mid u\in M_n,i\in\mathbb{N}\};\]

    \item each $u'$ is a successor of $u$ and each successor of $u$ (in
      $\M_n$) is a successor of $u'$ (in $\N_n$);


    \item for each $u$, the $u_i$, $i\in\mathbb{N}$ form an infinite path starting in $u$;\footnote{Here we need the assumption $d\geq 2$: in the constructed model $\N_n$, $u$ has at least two successors.}

    \item for each $u\in M_n$, the valuations of points in $N_n$ are defined as follows:
    \begin{align*}
      \val_N(u) & =\val_M(u)\cap \sigma \\
      \val_N(u') & =\{\widehat a\} \\
      \val(u_i) & = \begin{cases} \emptyset & \text{if $i\leq 2n$},\\ \val(u)\setminus \sigma & \text{if $i>2n$;}\end{cases}
    \end{align*}
%
%
    Intuitively, the last equation means that the non-$\sigma$ propositions are copied from $u$ sufficiently far away in the infinite path to be invisible by any $2n$-bisimulation.

  \end{itemize}
  The model $\N_n'$ is obtained analogously from $\M_n$. Using the
  game theoretic semantics of $\muML$ it is not hard to show that:

  \begin{clm} $\N_n\models h(\varphi,\sigma_\varphi)$ and $\N_n'\models h(\varphi,\sigma_{\varphi'})$. 
  \end{clm}

  Let $Z$ be a $(\sigma,n)$-bisimulation witnessing $\M_n\bis_\sigma^n \M_n'$. Based on the fact that, by construction of $\N_n,\N_n'$, no non-$\sigma$ proposition
  appears in the first $2n$ levels of $\N_n,\N_n'$, it is not difficult to verify that $Z'$ defined by taking
  \[Z'=Z\cup \{(u',v')\mid (u,v)\in Z\}\cup \{(u_i,v_i)\mid (u,v)\in Z,i\leq 2n\}\] 
  witnesses $\N_n\bis^{2n} \N_n'$. Hence, $\widehat\phi,\widehat\phi'$ are not $\ML$-separable.

  \medskip
  For ``$\Rightarrow$'', suppose that $\widehat\varphi,\widehat
  \varphi'$ are not $\ML$-separable. Then they are not $\ML_\tau$-separable for $\tau=\sig(\widehat\phi)\cup\sig(\widehat\phi')$.
  By Equivalence~\eqref{eq:non-sep vs n-bis}, for every
  $n\in\NN$, there are models
  $\M_n\models\widehat\varphi$, $\M_n'\models\widehat\varphi'$ with
  $\M_n\bis^{2n}_\tau\M_n'$. We assume that in $\M_n,\M'_n$:
  \begin{enumerate}

    \item from a point not satisfying $\widehat a$ we never reach a point satisfying $\widehat a$, 

    \item no non-$\sigma$ proposition is satisfied in the first $2n$ levels, and

    \item in even levels, every point satisfying $\widehat a$ has at most one successor satisfying $\widehat a$.

  \end{enumerate}
  We argue that these assumptions can be made without loss of generality. For (1), recall that we relativized our formulae to $\widehat a$. Hence, if we modify 
  $\M_n$ by making $\widehat a$ false in every subtree rooted at a
  point not satisfying $\widehat a$ it will still be a model of
  $\widehat\varphi$, and the same for $\M_n'$. Also the bisimulation is not affected by this modification. For~(2), observe that 
  non-$\sigma$ propositions appear in $\widehat\varphi,\widehat\varphi'$ only under a $\Diamond^*$-operator, so we can delay them arbitrarily
  far. For~(3), let $u$ be a point in an even level that has more than one successor satisfying $\widehat a$. It can be easily verified based on the structure of $\widehat\varphi,\widehat\varphi'$ that joining all successors into one (which has all the successors of the single points as successors) does not change $\M_n$ being a model of $\widehat\varphi$. We will from now on refer to this successor of $u$ satisfying $\widehat a$ with $u'$.
  After doing the same modification for $\M_n'$, we can again find a bisimulation. 
  
  We construct models $\N_n,\N_n'$ with universes $N_n,N_n'$ as follows: 
  %
  \begin{itemize}

    \item The universe $N_n$ of $\N_n$ is the smallest set $N$ such that $N$ contains
    the root of $\M_n$, and if $u\in N$ then $N$ contains also all successors of
    $u'$ (in $\M_n$) that satisfy $\widehat a$.

    \item For $c\in \sigma$ and $u\in N_n$, we have $\N_n,u\models c$
      iff $\M_n,u\models c$,

    \item For $c\in \sigma_\varphi$ and $u\in N_n$, we have $\N_n,u\models c$
      iff $\M_n,u\models \Diamond(\neg \widehat a\wedge \Diamond^*c)$,

  \end{itemize}
  The model $\N_n'$ is constructed analogously from $\M_n'$. Using the
  game theoretic semantics of $\muML$ it is not hard to show that:
  \begin{clm}
    $\N_n\models \varphi$ and $\N_n'\models\varphi'$.
  \end{clm}
  Let $Z$ be any $(2n,\tau)$-bisimulation witnessing
  $\M_n\bis^{2n}_\tau\M_n'$. It is routine to verify that the restriction
  of $Z$ to $N_n\times N_n'$ witnesses $\N_n\bis^{n}_\sigma\N_n'$. Hence, $\varphi,\varphi$ are not $\ML_\sigma$-separable.
\end{proof}

%
\subsection{Proof of Upper Bound in Theorem~\ref{thm:ternary complexity}}\label{sec:upper}
We show that over models of outdegree at most $d$, $\ML$-separability of fixpoint formulae can be solved in doubly exponential time.
Let us start with establishing a technical but useful fact. For every language of $d$-ary trees $L\subset \TT^d$ denote the language:
\[
  \bisQuotient(L) = \{\M\in\TT^d\ |\ \text{there is $\N\in L$ and a functional bisimulation $Z:\N\bisFun{}\M$}\}
\]
of bisimulation quotients of trees from $L$.
\begin{prop}\label{prop:quotient automaton}
  For every NPTA $\A$ over $\TT^d$, an NPTA $\B$ in trash normal form recognizing $\bisQuotient(\LL(\A))$ can be computed in time exponential in the size of $\A$.
\end{prop}

\begin{proof}
Fix an NPTA $\A=(Q,\Sigma,q_I,\delta,\rank)$. If necessary, we modify $\A$ so that no transition $S=(p_1,...,p_k)$ contains two occurrences of the same state. 
For every $\M\in\TT^d$, we characterize existence of $d$-ary $\N\models\A$ with $\N\bisFun{}\M$ with the following parity game $\game_\bisQuotient(\M,\A)$.
The game has the set $M\times Q$ as positions. The pair $(\point_I,q_I)$ consisting of the root $\point_I$ of $\M$ and $q_I$ is the initial position. From a position $(\point,q)$ first $\eve$ chooses $S\in\delta(q,\val(c))$ and a surjective map $h:S\to\{\point_1,...,\point_k\}$ where $\{\point_1,...,\point_k\}$ is the set of children of $\point$. Then $\adam$ responds with a choice of $p\in S$ and the next round starts in position $(h(p),p)$. The game is a parity game: the ranks are inherited from $\A$ in the sense that the rank of $(\point,q)$ equals $\rank(q)$. 
The game is designed so that:
\begin{align}\label{eq:game-quotient}
  \text{$\eve$ wins $\game_\bisQuotient(\M,\A)$}  \iff \M\in\bisQuotient(\LL(\A))
\end{align}
for every $\M\in\TT^d$. We prove~\eqref{eq:game-quotient} now.

We start with the implication ``$\Rightarrow$'' from left to right. Assume a winning strategy $\strategy$ for $\game_\bisQuotient(\M,\A)$. We construct $\N\in\TT^d$ accepted by $\A$ together with $Z:\N\bisFun{}\M$.
The universe $N$ consists of all finite $\strategy$-plays. There is an edge $\pi\to\pi'$ from $\pi$ to $\pi'$ if $\pi'$ extends $\pi$ with one move of $\eve$ followed by a response of $\adam$. This means that the outdegree of $\pi$ equals the size of the transition $S$ chosen by $\strategy$ as a response to $\pi$. In particular, $\N$ is $d$-ary. Define $Z:N\to M$ such that $Z(\pi)$ is the point component $\point$ from the last configuration $(\point,q)$ of $\pi$. We complete the definition of $\N$ by putting $\val^\N(\pi)=\val^\M(Z(\pi))$ for all $\pi\in N$.

We claim that the function $Z$ is a bisimulation. The base condition $\baseCond$ follows immediately from the definition of $\val^\N$. To prove the $\backCond$ and $\forthCond$ conditions assume $\pi\in\N$.
Denote $Z(\pi)=\point$, let $\point_1,...,\point_k$ be the children of $\point$ and $h:S\to\{\point_1,...,\point_k\}$ be the move chosen by $\strategy$ as a response to $\pi$.

To show the $\forthCond$ condition assume $\pi\to\pi'$. Existence of the edge $\pi\to\pi'$ implies that $\pi'$ is of shape $\pi\cdot(h(p),p)$ for some $p\in S$. Let $\point'=h(p)$. Since $\point\to\point'$ and $Z(\pi')=\point'$ this completes the argument of the $\forthCond$ condition.
Towards the $\backCond$ condition assume $\point\to\point'$ for some $\point'$. Since $\point'$ is a child of $\point$ and $h$ is surjective, there is $p\in S$ such that $h(p)=\point'$. Then $\pi'=\pi\cdot(\point',p)$ is a $\strategy$-play such that $Z(\pi')=\point'$ and $\pi\to\pi'$. Hence, $\pi'$ witnesses the $\backCond$ condition.

It remains to construct an accepting run $\rho:N\to Q$. For each $\pi\in N$ we define $\rho(\pi)$ as the state component of the last configuration in $\pi$. This $\rho$ is consistent with $\delta$. To show this assume $\pi$ with children $\pi_1,...,\pi_k$. Denote $\pi_i=\pi\cdot(\point_i,p_i)$ for each $i\leq k$ and let $(\point,q)$ be the last configuration in $\pi$. There is $S\in\delta(q,\val^\M(\point))$ such that $S=\{p_1,...,p_k\}$. Since $\rho(\pi)=q$, $\rho(\pi_i)=p_i$ for every $i\leq k$ and $\val^\N(\pi)=\val^\M(\point)$, the transition $S$ is legal in $\pi$. To see that $\rho$ is accepting assume an infinite path $\pi_1,\pi_2,...$ in $\N$ and for each $i$ let $(\point_i,q_i)$ be the last configuration of $\pi_i$ (so in particular $\rho(\pi_i)=q_i$). We need to show that the sequence $q_1q_2...$ of states satisfies the parity condition. This is true because each $\pi_i$ is a $\strategy$-play and therefore so is their infinite limit $(\point_1,q_1)(\point_2,q_2)...=\pi$. This completes the proof of the implication ``$\Rightarrow$'' in~\eqref{eq:game-quotient}.

Let us prove the other implication ``$\Leftarrow$'' in~\eqref{eq:game-quotient}. Assume $\N\in\TT^d$ such that $\N\models\A$ and $Z:\N\bisFun{}\M$, and let $\rho:N\to Q$ be an accepting run witnessing $\N\models\A$. We construct a winning strategy $\strategy$ for $\eve$ in $\game_\bisQuotient(\M,\A)$. The constructed strategy preserves as an invariant that for every $\strategy$-play $\pi=(\point_1,q_1)...(\point_l,q_l)$ there is a path $\altpoint_1...\altpoint_l$ in $\N$ with $Z(\altpoint_i)=\point_i$ and $\rho(\altpoint_i)=q_i$ for each $i\leq l$. The invariant holds in the initial position $(\point_I,q_I)$. To define moves dictated by $\strategy$ assume a play $\pi=(\point_1,q_1)...(\point_l,q_l)$ and a path $\altpoint_1...\altpoint_l$ from the invariant. Let $\point_1',...\point_k'$ be the children of $\point_l$. We define an $\eve$'s move $h:S\to\{\point_1',...,\point_k'\}$ dictated by $\strategy$ as a response to $\pi$.

Let $S$ be the transition chosen by $\rho$ in $\altpoint_l$. By assumption no state of $\A$ appears more than once in a single transition $S$ of $\delta$. This means that $\rho$ is bijective between the set of children of $\altpoint_l$ and $S$: for each $p\in S$ there exists a unique child $\altpoint^p$ of $\altpoint_l$ such that $\rho(\altpoint^p)=p$. We set $h(p)=Z(\altpoint^p)$ for every $p\in S$. By the $\baseCond$ condition $\altpoint_l$ and $\point_l$ have the same color so to show that this $h$ is a legal move for $\eve$ it suffices to show that it is surjective. By the $\backCond$ condition for every child $\point'$ of $\point_l$ there is some child $\altpoint'$ of $\altpoint_l$ with $Z(\altpoint')=\point'$. Hence, $h(\rho(\altpoint'))=\point'$ which proves surjectivity of $h$. Moreover, the invariant is preserved: if $\pi'$ extends $\pi$ by $\adam$'s response $(h(p),p)$ to $h$ then $Z(\altpoint^p)=h(p)$ and $\rho(\altpoint^p)=p$ and so we extend $\altpoint_1...\altpoint_k$ with $\altpoint_{k+1}=\altpoint^p$.

To see that the strategy $\strategy$ is winning observe that, thanks to the invariant, for every infinite $\strategy$-play $\pi=(\point_1,q_1)(\point_2,q_2)...$ there exists an infinite path $\overline{\altpoint}=\altpoint_1\altpoint_2...$ in $\N$ such that $Z(\altpoint_i)=\point_i$ and $\rho(\altpoint_i)=q_i$ for all $i$. Thus, $\eve$ wins $\pi$ because $\rho$ is accepting. This completes the proof of~\eqref{eq:game-quotient}.

Using~\eqref{eq:game-quotient} we prove Proposition~\ref{prop:quotient automaton}. It suffices to construct an automaton $\B$ which accepts $\M$ iff $\eve$ wins $\game_\bisQuotient(\M,\A)$. By positional determinacy of parity games, to check the latter it suffices to look at \emph{positional} winning strategies for $\eve$ in $\game_\bisQuotient(\M,\A)$.
Using standard techniques we construct in exponential time an automaton $\B$ in trash normal form which, given $\M$, verifies existence of such a strategy.
Since this step follows the classical constructions from Theorem~\ref{thm:munpta} and Corollary~\ref{cor:muML to NPTA}, below we only give a brief sketch.
\begin{enumerate}
  \item A positional strategy $\strategy$ for $\eve$ in $\game_\bisQuotient(\M,\A)$ is formally a map which sends every position $(v,q)$ of the game to $\eve$'s answer, that is, to an appropriate map $h$ from some $S\subset Q$ to $v$'s children. We encode such $\strategy$ as a coloring $\lambda_\strategy:M\to\powerset{Q\times Q}$ in the following way. Assume $\point$ in $\M$ with children $\point_1,...,\point_k$ and for each state $q$ denote $\strategy(\point,q)$ by $h_q:S\to\{\point_1,...,\point_k\}$. For each $i\leq k$ we set $\lambda_\zeta(\point_i)=\{(q,p)\ |\ h_q(p)=\point_i\}$. 
  \item Construct $\B^+$ which inputs a model $\M$ together with a labelling $\lambda:M\to\powerset{Q\times Q}$ and accepts iff $\lambda$ encodes a winning strategy for $\eve$ in $\game_\bisQuotient(\M,\A)$. The key part in the construction of $\B^+$ is to check that every infinite play consistent with the strategy encoded by $\lambda$ satisfies the parity condition. This is done effectively by using efficient determinization of word automata, in exactly the same way as in Theorem~\ref{thm:munpta}.
  \item Obtain $\B'$ which recognizes $\bisQuotient(\LL(\A))$ by projecting out the additional colors $\powerset{Q\times Q}$ from $\B^+$.
  \item Finally, we bring $\B'$ to $\B$ in trash normal form. This step only requires single exponential time, for the same reason as in Corollary~\ref{cor:muML to NPTA}: $\B'$ has only polynomially many ranks.
\end{enumerate}
This completes the proof of Proposition~\ref{prop:quotient automaton}.
\end{proof}
With the help of Proposition~\ref{prop:quotient automaton} we complete the proof of
Theorem~\ref{thm:ternary complexity}. Fix $d$ and $\muML$-formulae
$\phi$ and $\phi'$ over signature $\sigma$. As discussed in Section~\ref{sec:foundations}, it is enough to look for $\ML$ separators which only use symbols from $\sigma$. By Equivalence~\eqref{eq:non-sep vs n-bis}, it suffices to check if $\phi$ and $\phi'$ are jointly $\bis_\sigma^n$-consistent over $\TT^d$ for every $n$. However, unlike with definability or in the binary case, we cannot conclude joint $\cong_\sigma^n$-consistency from joint $\bis_\sigma^n$-consistency. Instead, we use Proposition~\ref{prop:quotient automaton} to directly decide joint $\bis_\sigma^n$-consistency for all $n$.
For a language $L\subset\TT^d$, define the language:
\[
  \mathsf{QPL}(L) = \{\N\in\TT^d\ |\ \text{there is $\M\in L$, finite prefix $\M_0$ of $\M$ and $Z:\M_0\bisFun\N$}\}
\]
of finite $d$-ary trees which are bisimulation quotients of finite prefixes of models from $L$.
By Proposition~\ref{prop:quotient automaton} and the closure properties of parity automata, for every $\A$ one can construct in exponential time an automaton $\B$ recognizing $\mathsf{QPL}(\LL(\A))$.

We prove the upper bound from Theorem~\ref{thm:ternary complexity}. Using Corollary~\ref{cor:muML to NPTA} compute automata $\A,\A'$ equivalent to $\phi,\phi'$. Use Proposition~\ref{prop:quotient automaton} to compute automata $\B,\B'$ in trash normal form which recognize $\mathsf{QPL}(\LL(\A))$ and $\mathsf{QPL}(\LL(\A'))$. Recall that any two trees are bisimilar iff they have isomorphic bisimilarity quotients. It follows that $\phi,\phi'$ admit an $\ML^n$-separator over $\TT^d$ iff $\A,\A'$ are jointly $\bis_\sigma^n$-consistent iff $\B,\B'$ are jointly $\cong^n_\sigma$ consistent. By Proposition~\ref{prop:automata step}, the latter condition holds for all $n\in\NN$ iff it holds for $n=|\B|\times|\B'|+1$ and this can be tested in time polynomial in $(|\Sigma|+|\Sigma'|)\times(|\B|+|\B'|)^d$ where $\Sigma=\powerset{\sig(\phi)}$ is the alphabet of $\A$ and of $\B$, and $\Sigma'=\powerset{\sig(\phi')}$ is the alphabet of $\A'$ and of $\B'$.
Since $\A,\A'$ are exponential, and $\B,\B'$ are doubly exponential in the size of $\phi,\phi'$, this gives the upper bound from Theorem~\ref{thm:ternary complexity}.

We conclude this subsection with an observation that in case an $\ML$-separator exists, we can also compute one in elementary time.
 
\begin{restatable}{thm}{thmternaryconstruction}\label{thm:ternary construction}
  If $\varphi,\varphi'$ are $\ML$-separable over $\TT^d$,
  $d\geq 3$, then one can compute an $\ML$-separator in time
  triply exponential in $|\varphi|+|\varphi'|$.
\end{restatable}

\begin{proof} It follows from the upper bound proof of
  Theorem~\ref{thm:ternary complexity} that if $\varphi,\varphi'$
  admit an $\ML$-separator, then they admit one of modal depth
  bounded doubly exponentially in $|\varphi|+|\varphi'|$. Observe that over the signature
  of $\varphi$ and $\varphi'$ there are only triple exponentially many trees of fixed
  outdegree $d$ and double exponential depth, and that each such tree
  is characterized (up to bisimulation) by a modal formula of triply exponential size.
  The sought separator is then the disjunction of all such formulae
  consistent with $\varphi$.
\end{proof}

\section{Case Study: Graded Modalities}\label{sec:graded}

In this section we apply our techniques and results to the case with \emph{graded
modal operators}. We extend $\muML$ with formulae of the
shape $\Diamond_{\geq g}\psi$ and $\Box_{\geq g}\psi$, where $g\in\NN$ is a
natural number called a  \emph{grade}. 
We denote with $\grML$ and $\mugrML$ the extension of $\ML$ and $\muML$, respectively, with such graded modalities.
Intuitively, $\Diamond_{\geq g}\psi$ is true in a
point $w$ if $w$ has at least $g$ successors satisfying $\psi$ and
dually, $\Box_{\geq g}\psi$ is true in $w$ if all but at most $g$
successors satisfy $\psi$~\cite{DBLP:journals/ndjfl/Fine72,DBLP:journals/corr/abs-1910-00039}. Formally, the semantics of $\mugrML$ is given by semantic games. These games extend the semantic games used to define the semantics of $\muML$.
As in the case with no grades, the game for a formula $\phi$ and a model $\M$ has $M\times\SubFor(\phi)$ as the set of positions, the winning (parity) condition is defined as in the classical case, and so are the moves for all the positions with topmost connective other than the graded modalities. In the classical game, from $(\point,\Diamond\theta)$ $\eve$ chooses a child $\point'$ of $\point$ and the next position is $(\point',\theta)$.
In $(\point,\Diamond_{\geq k}\theta)$, first $\eve$ chooses a subset $\point_1,...,\point_l$ of children of $\point$ of size $l\geq k$, then $\adam$ chooses one of these children $\point_i$ and the next round starts at $(\point_i,\theta)$. Dually, in $(\point,\Box_{k\geq}\theta)$ first $\eve$ picks a subset $\point_1,...,\point_l$ of $\point$'s children of size $l\leq k$, then $\adam$ responds with a choice of some $\point'$ not in $\point_1,...,\point_l$ and the next position is $(\point',\theta)$. Note that without loss of generality $\eve$ always chooses subsets of minimal size $l=k$ in $(\point,\Diamond_{\geq k}\theta)$, and of maximal size $l=\min(k,\text{number of $v$'s children})$ in $(\point,\Box_{k\geq}\theta)$. We thus tacitly assume such optimality when considering strategies for her.
%
%
Clearly, for any $d\in\NN$, the class $\TT^d$ is
$\mugrML$-definable by the formula $\theta_d = \nu
x.(\Box x\wedge\Box_{\geq d}\bot)$, which is an additional motivation to
study $\grML$ and
$\mugrML$.

We use tools and techniques from the previous sections to solve separability and definability problems in the graded setting. These questions split into two categories, depending on wether grades are allowed in the separating/defining formula or only in the input.
Before we inspect these two cases in turn, let us discuss two limitations of this inquiry.
First, we only look at the decision problems: deciding existence of a suitable separator/definition. The problem of efficiently \emph{constructing} one when it exists is a nontrivial task, at least in the case with no grades in the separator, and goes far beyond the scope of this paper. Second, in this section we assume that the grades in formulae are represented in unary. This is admittedly a serious limitation, but is needed for an exponential translation from $\mugrML$ to automata over $\TT^d$, generalizing Theorem~\ref{thm:munpta} (and its strengthening Corollary~\ref{cor:muML to NPTA}).
It is likely that this restriction can be lifted by working with a more succinct automaton model suitable for efficiently storing grades, following the approach of~\cite{DBLP:conf/cade/KupfermanSV02}. Adapting our proofs to this setting is not immediate, however.

\subsection{With Grades in the Separator}
We first investigate $\grML$-definability and $\grML$-separability of $\mugrML$-formulae.
We start by introducing a graded variant of bisimulations called \emph{graded bisimulations}, which characterize the expressive power of $\grML$~\cite{DBLP:journals/sLogica/Rijke00}. A relation $Z$ between models is a graded
bisimulation if it satisfies ($\baseCond$) and \emph{graded} variants
of the ($\backCond$) and ($\forthCond$) conditions of bisimulations.
The graded ($\forthCond$) condition says that if  $\point Z\altpoint$
then for every $k\in \NN$ and pairwise different children
$\point_1,...,\point_k$ of $\point$, there are pairwise different children
$\altpoint_1,...,\altpoint_k$ of $\altpoint$ satisfying $\point_i
Z\altpoint_i$ for all $i\leq k$. The graded ($\backCond$) condition is
symmetric. $Z$ is a \emph{$g$-graded bisimulation} if the graded
($\forthCond$) and ($\backCond$) conditions need to be satisfied only for
$k\leq g$.  We denote with $\M\bis_\graded\M'$ (resp., $\M\bis_{\graded(g)}\M'$)
the fact that there is a graded bisimulation (resp., a $g$-graded
bisimulation) between $\M$ and $\M'$ that relates their roots.
Variants with bounded depth $n$ and/or given
signature $\sigma$ are defined and denoted as expected.

\begin{restatable}{thm}{thmgraded}\label{thm:graded separability complexity}
  The problems of $\grML$-definability and $\grML$-separability of $\mugrML$-formulae are both \ExpTime-complete.
\end{restatable}

The lower bounds follow by the usual reduction from satisfiability. We thus focus on the upper bound in the more general case of separability. Similarly
to the non-graded case, we establish first a model-theoretic
characterization based on graded bisimilarity.

\begin{restatable}{lem}{lemgradedml}\label{lem:graded model theory}
  For every $\phi,\phi'\in\mugrML$ with maximal grade $g_\text{max}$, signature $\sigma$, and $n\in\NN$, the following are equivalent:
\begin{enumerate}

  \item\label{it:graded model theory 1} $\phi,\phi'$ are not
    $\grML^n_\sigma$-separable (over all models).

  \item\label{it:graded model theory 2} $\phi,\phi'$ are jointly
  $\bis^n_{\graded,\sigma}$-consistent (over all models).

  \item\label{it:graded model theory 3} $\phi,\phi'$ are jointly
  $\cong^n_{\sigma}$-consistent (over all models).

  \item\label{it:graded model theory 4} $\phi,\varphi'$ are jointly
    $\cong^n_\sigma$-consistent over $\TT^d$ for
    $d=g_\text{max}\times(|\phi|+|\phi'|)$.

\end{enumerate}
\end{restatable}
Using Lemma~\ref{lem:graded model theory} one can solve
$\grML$-separability of $\mugrML$-formulae in exponential time,
following the approach described in
Section~\ref{sec:foundations}. More precisely, given $\phi,\phi'$, we
construct NPTA $\Amc,\Amc'$ equivalent to $\phi,\phi'$ over $d$-ary
trees, $d$ as in Lemma~\ref{lem:graded model theory}, and decide
whether $\Amc,\Amc'$ are jointly
$\cong_\sigma^n$-consistent over $\TT^d$ for all $n$ via
Proposition~\ref{prop:automata step}. It therefore suffices to prove the central Lemma~\ref{lem:graded model theory}. 

\begin{proof} 
  We show the implications $1\Rightarrow 2\Rightarrow
  3\Rightarrow 4 \Rightarrow 1$ in turn. The implication $4\Rightarrow
  1$ is immediate. 

  For $1\Rightarrow 2$, suppose
  $\phi,\phi'$ are not $\grML^n_\sigma$-separable. Hence, for every
  $g\in\NN$ there is a pair of models $\M_g\models\phi$ and
$\M_g'\models\phi'$ with $\M_g\bis_{\mathsf{grd}(g),\sigma}^n\M_g'$. One can encode
with an $\FO$-sentence $\theta$ that two models $\M$ and $\M'$ are
depth-$n$ trees, $\M$ is a prefix of some $\M_+\models\phi$ and $\M'$
of some $\M_+'\models\phi'$. If $Z$ is a fresh binary symbol, then it
is also possible to encode with an infinite set $T$ of $\FO$-sentences that $Z$
is a graded bisimulation between $\M$ and $\M'$: for every grade $g$ we take a separate $\FO$-sentence. Every finite fragment
of $\{\theta\}\cup T$ only mentions finitely many grades and hence by
assumption is satisfiable. Thus, by compactness of $\FO$, the entire
$\{\theta\}\cup T$ is satisfiable. This gives us
$\M\bis_{\graded,\sigma}^n\M'$
with extensions $\M_+\models\phi$ and $\M_+'\models\phi'$. 

For $2\Rightarrow 3$, fix witnesses $\M,\M'$ of joint
$\bis_{\graded,\sigma}^n$-consistency. That is,
$\M\bis_{\graded,\sigma}^n\M'$ and there are extensions $\M_+,\M_+'$
of $\M,\M'$ with $\M_+\models\phi$ and $\M_+'\models\phi'$. 
By the
L\"{o}wenheim-Skolem property of $\FO$ we may assume that the $n$-prefixes of both models
are at most countable. It remains to apply the known fact that countable trees $\N$ and
$\N'$ satisfy $\N\bis_\graded\N'$ iff $\N$ and $\N'$ are isomorphic.
For the sake of completeness, we add a brief justification of this
latter statement. Assume $\altpoint\in\N$ and $\altpoint'\in\N$ with
respective children $\{\altpoint_1,\altpoint_2,...\}=W$
and $\{\altpoint_1',\altpoint_2',...\}=W'$ such that
$\altpoint\bis_\graded\altpoint'$. For every
$\bis_\graded$-equivalence class $X$ of $W$ the
corresponding equivalence class $X'=\{\altpoint'\in W'\ |\ \exists_{w\in X}.\ \altpoint\bis_\graded\altpoint'\}\subset W'$ has the same
cardinality as $X$. This is immediate for finite $X$, and for infinite
$X$ it follows because in countable models every two infinite subsets
have the same cardinality. This allows us to inductively pick a
\emph{bijective} subrelation $Z$ of $\bis_\graded$ between $\N$ and
$\N'$ which is still a graded bisimulation.

For $3\Rightarrow 4$, fix witnesses $\M,\M'$ of joint
$\cong_{\sigma}^n$-consistency. That is,
$\M\cong_{\sigma}^n\M'$ with respective extensions $\M_+,\M_+'$
such that $\M_+\models\phi$ and $\M_+'\models\phi'$.
We trim $\M_+$ and $\M_+'$ so that the outdegree becomes at most $d$.
Without loosing generality we assume that the prefixes of $\M_+$ and
$\M_+'$ are not only isomorphic but identical. 

Let $\strategy$ and $\strategy'$ be positional winning strategies for $\eve$ in the semantic games $\game(\M_+,\phi)$ and $\game(\M_+',\phi')$. We take a submodel $\M_0\models\phi$ of $\M_+$ as follows. In the $n$-prefix we take the root and all $\Diamond$-witnesses for $\strategy$ and $\strategy'$: points chosen by $\strategy$ or $\strategy'$ in a position whose formula component starts with $\Diamond_{\geq g}$. In the rest of the model we only take $\Diamond$-witnesses for $\strategy$. A submodel $\M_0'$ of $\M_+'$ is defined symmetrically. It follows that $\M_0\models\phi$ and $\M_0'\models\phi'$. 

Recall that $g$ is the maximal grade appearing in $\phi$ and $\phi'$.
Since the respective sets of positions of $\game(\M_+,\phi)$ and
$\game(\M_+',\phi')$ are $M_+\times\SubFor(\phi)$ and
$M_+'\times\SubFor(\phi')$, for every point $\point$ there are at most
$g\times|\phi|$ $\Diamond$-witnesses chosen by $\strategy$ from a
position which has $\point$ on the first coordinate. Consequently, the
outdegree of $\M_0$ and $\M_0'$ is not greater than
$d=g\times(|\phi|+|\phi'|)$. This proves Lemma~\ref{lem:graded model theory}.
\end{proof}

\subsection{Without Grades in the Separator}
We now consider the other case with no grades allowed in the separator/definition. Interestingly, this change in the setting increases the complexity of separability, but not definability. 

We start with the easier $\ML$-definability of $\mugrML$-formulae, following a strategy pursued for deciding $\ML$-definability of $\grML$-formulae~\cite{separatingcounting}. As an intermediate step, and because it is interesting in its own right, we also look at $\muML$-definability of $\mugrML$-formulae. It turns out that if a $\mugrML$-formula is $\muML$-definable then this is already witnessed by the most naive candidate for a definition. Let \emph{flattening} of  a given $\phi\in\mugrML$, denoted $\mathsf{flat}(\phi)$, be a $\muML$-formula obtained from $\phi$ by replacing every $\Diamond_{\geq k}$ with $\Diamond$ and every $\Box_{\geq k}$ with $\Box$.

\begin{lem}\label{lem:graded flattening}
  A $\mugrML$-formula $\phi$ is equivalent to a $\muML$-formula iff it is equivalent to $\mathsf{flat}(\phi)$.
\end{lem}
\begin{proof}
We prove the only nontrivial implication from left to right. Assume $\phi$ is equivalent to some $\muML$-formula. In particular, this means that $\phi$ is invariant under bisimulation. Consider a model $\M$. Let $\M^+$ be obtained from $\M$ by duplicating subtrees so that every node has infinitely many bisimilar siblings. We have:
\begin{align*}
  \M\models\phi \hspace{0.1cm} &\iff \hspace{0.1cm} \M^+\models\phi\\
                \hspace{0.1cm} &\iff \hspace{0.1cm} \M^+\models\mathsf{flat}(\phi)\\
                \hspace{0.1cm} &\iff \hspace{0.1cm} \M\models\mathsf{flat}(\phi).
\end{align*}
The first and third equivalence follow from bisimulation invariance. The second one follows from the observation that in $\M^+$ the semantics of $\Diamond_{\geq k}$ coincides with that of $\Diamond$, and similarly with $\Box_{\geq k}$ and $\Box$.
\end{proof}

Lemma~\ref{lem:graded flattening} reduces $\muML$-definability of a given $\mugrML$-formula $\phi$ to checking the validity of $\phi\iff\mathsf{flat}(\phi)$. Moreover, $\phi$ is $\ML$-definable iff (i) it is equivalent to its flattening $\mathsf{flat}(\phi)$ and (ii) $\mathsf{flat}(\phi)$ is $\ML$-definable. Since the validity of $\phi\iff\mathsf{flat}(\phi)$ can be decided in exponential time~\cite{DBLP:conf/ijcai/CalvaneseGL99,DBLP:conf/cade/KupfermanSV02} and the latter (ii) can be decided via Theorem~\ref{thm:otto main}, we obtain the following complexities.

\begin{thm}
  $\muML$-definability and $\ML$-definability of $\mugrML$-formulae are \ExpTime-complete.
\end{thm}

As mentioned, without grades in the separator general separability has higher complexity than definability.

\begin{restatable}{thm}{thmhalfgraded}\label{thm:half-graded complexity}
  $\ML$-separability of $\mugrML$-formulae is \TwoExpTime-complete.
\end{restatable}

\begin{proof}
For the lower bound, we reduce $\ML$-separability of $\muML$-formulae
over $\TT^3$. Since
the former problem is \TwoExpTime-hard by Theorem~\ref{thm:ternary
complexity}, so is the latter. The reduction relies on the formula $\theta_3$
defining $\TT^3$. It is an immediate consequence of the definitions that for all
$\muML$-formulae $\varphi,\varphi'$ and $\psi\in \ML$,  
\[\text{$\psi$ is an $\ML$-separator of $\varphi,\varphi'$ over
$\TT^3$\quad iff\quad $\psi$ is an $\ML$-separator of $\varphi\wedge
\theta_3,\varphi'\wedge \theta_3$}.\] 
%

Towards the upper bound, assume $\phi,\phi'\in\mugrML$. As usual, we only need to look for separators among formulae over $\sigma=\sig(\phi)\cup\sig(\phi')$. By an immediate generalization of the Equivalence~\eqref{eq:non-sep vs n-bis}, $\phi,\phi'$ are not $\ML_\sigma^n$-separable iff they are jointly $\bis_\sigma^n$-consistent.
A standard game-theoretic argument similar to the one in the proof of $3\Rightarrow 4$ in Lemma~\ref{lem:graded model theory} shows that $\phi,\phi$ are jointly $\bis_\sigma^n$-consistent over all models iff they are jointly $\bis_\sigma^n$-consistent over $\TT^d$, where
$d=g\times(|\phi|+|\phi'|)$ and $g$ is the greatest grade occurring in $\phi,\phi'$. Summing up, $\phi,\phi'$ are not $\ML$-separable iff they are jointly $\bis_\sigma^n$-consistent over $\TT^d$ for all $n\in\NN$.
To decide the latter we proceed as in the upper bound proof of Theorem~\ref{thm:ternary complexity}. We first construct NPTA $\Amc,\Amc'$ equivalent to $\varphi,\varphi'$ over $d$-ary trees via (an analogue for $\mugrML$ of) Corollary~\ref{cor:muML to NPTA}. Then, using Proposition~\ref{prop:quotient automaton} we construct $\B$ recognizing $\mathsf{QPL}(\LL(\A))$ and $\B'$ recognizing  $\mathsf{QPL}(\LL(\A'))$ and apply Proposition~\ref{prop:automata step} to decide whether $\B,\B'$ are jointly $\cong_\sigma^n$-consistent for all $n$.
\end{proof}

\section{Final Discussion}\label{sec:conclusion}

We have presented an in-depth study of modal separation of
$\muML$-formulae over different classes of structures. For us, the
most interesting results are the differences that are obtained over
classes of bounded outdegree for different bounds $d=1$, $d=2$, $d\geq
3$. It is worth noting that without much effort our (upper bound) results can be transferred to other interesting classes of models. First, it is known that $\muML$ has the tree model property and the finite model property. Due to this, for every $\psi,\phi,\phi'\in\muML$ it follows that $\psi$ is a separator of $\varphi,\varphi'$ (over all models) iff $\psi$ is a separator of $\varphi,\varphi'$ over trees iff $\psi$ is a separator of $\varphi,\varphi'$ over finite models.
Thus, $\ML$-separability coincides over all these classes. Moreover, with the help of the $\muML$-formula $\theta_\infty$ from Example~\ref{exa:separability} we can show that
$\ML$-separability over finite trees reduces to $\ML$-separability
over all models. More formally: 
\begin{lem}\label{lem:reduction-WF}
  Let $\varphi,\varphi'\in\muML$ and $\psi\in\ML$. Then
$\psi$ is an $\ML$-separator of $\varphi,\varphi'$
      over finite trees 
      iff $\psi$ is an $\ML$-separator of
      $\varphi\wedge\neg\theta_\infty,\varphi'\wedge
      \neg \theta_\infty$. This is also true inside $\TT^d$, for 
      $d\in\NN$.
\end{lem}
Via similar reductions, our upper bound results do also transfer to
infinite words and \emph{ranked} trees. 

We find also worth noting that the construction of $\ML^n$-uniform consequences of $\muML$ formulae might be of independent interest. For example, the construction can easily be adapted to construct $\ML^n_\sigma$-uniform consequences, which in turn provides an alternative route to computing (uniform) interpolants of $\ML$-formulae.

We discuss a few more aspects and possible extensions of our work.

\subsection{Craig Separability.} Inspired by Craig interpolation, it is natural to also consider the following variant of separability. Given $\phi,\phi'\in\muML$, we say that a formula $\psi\in\ML$ is a \emph{Craig $\ML$-separator} of $\phi,\phi'$ if $\sig(\psi)\subseteq \sig(\phi)\cap \sig(\phi')$ and $\psi$ separates $\phi,\phi'$. \emph{Craig $\ML$-separability of $\muML$-formulae} is then the problem of deciding whether given $\phi,\phi'\in\muML$ admit a Craig $\ML$-separator. Interestingly, $\ML$-separability and Craig $\ML$-separability coincide over all models, over $\TT^1$, and over $\TT^2$.
\begin{prop}\label{prop:sep=Craig}
  Over the class of all models, as well as over $\TT^1$ and $\TT^2$, $\varphi,\varphi'\in\muML$ admit an $\ML$-separator iff they admit a Craig $\ML$-separator.
\end{prop}
\begin{proof}
The ``if''-direction is trivial. Clearly, any Craig $\ML$-separator is also an $\ML$-separator.

For ``only if'' we will use the fact that $\ML$ enjoys the CIP over each of the considered classes $\ClassMod$: over all models and over $\TT^1$~\cite{10.1007/BFb0059541}, and over $\TT^2$ by Proposition~\ref{prop:Craig interpolation bin}.
Suppose there is an $\ML$-separator $\psi$ of $\varphi,\varphi'$ and let $n$ be its modal depth. Let $\theta,\theta'$ be respective $\ML^n$-uniform consequences of $\phi,\phi'$, which exist by Proposition~\ref{prop:L-uniform construction} or Proposition~\ref{prop:L-uniform construction1} depending on the class $\ClassMod$. Note that in both cases we have $\sig(\theta)\subset\sig(\phi)$ and $\sig(\theta')\subset\sig(\phi')$. Moreover, since $\psi\in\ML^n$, we have $\phi\models\theta\models\psi\models\neg\theta'\models\neg\phi'$. Thus, a Craig interpolant for $\theta,\neg\theta'$ is a Craig separator for $\phi,\phi'$.
\end{proof}
Clearly, Proposition~\ref{prop:sep=Craig} cannot hold over ternary (and higher arity) trees, since $\ML$ does not even enjoy the CIP over $\TT^3$ for $d\geq 3$. It is also interesting to note that the CIP alone is not sufficient for Proposition~\ref{prop:sep=Craig}: there are classes $\ClassMod$ where $\ML$ enjoys the CIP but $\ML$-separability and Craig $\ML$-separability of $\muML$-formulae do not coincide. Fix a proposition $\atProp\in\Propositions$ and consider:
\[
\ClassMod =
\left\{
  \M\in\TT^1 \;\middle|\;
  \begin{aligned}
  & |M|<\infty \text{ and $\atProp$ is false everywhere, or}\\
  & |M|=\infty \text{ and $\atProp$ is true everywhere}
  \end{aligned}
\right\}.
\]
The formula $\theta_\infty$ is equivalent to $\atProp$ over $\ClassMod$. However, $\sig(\theta_\infty)=\emptyset$ and no modal formula over the empty signature separates $\theta_\infty,\neg\theta_\infty$ over $\ClassMod$. Hence, $\ML$-separability and Craig $\ML$-separability do not coincide over $\ClassMod$. On the other hand, it is not difficult to derive the CIP of $\ML$ over $\ClassMod$ from the CIP of $\ML$ over $\TT^1$.

\subsection{Other Size Measures.}
Throughout the paper we used the simplest possible measure of formula size: the length of a formula written as a string. Alternative more succinct measures, such as the number of non-isomorphic subformulae (\emph{DAG-size}), are also interesting. Thus, a natural question is to what extent our results depend on the choice of size measure. In principle, using a more succinct measure makes the problems of definability and separability harder. However, all our decision procedures, with an exception of Theorem~\ref{thm:ml-interpolant-existence}, are automata-based. Consequently, these procedures carry over with unchanged complexity to any size measure for which the translation from logic to nondeterministic automata has the same complexity as in Corollary~\ref{cor:muML to NPTA}. In the remaining case of Theorem~\ref{thm:ml-interpolant-existence} a sufficient assumption is that the modal depth of a formula is at most polynomial in its size. Both the mentioned assumptions are arguably modest. 

A place where the choice of size measure matters a little more is the \emph{construction} of modal definitions and separators. In the cases of unrestricted, unary ($\TT^1$), and high outdegree models ($\TT^d$ for $d\geq 3$) the constructed formulae have DAG-size essentially the same as size: doubly, singly, and triply exponential, respectively. Interestingly, however, in the binary case $\TT^2$ our formulae have only singly exponential DAG-size, which is easily seen to be optimal and contrasts with their doubly exponential size.
This demonstrates that the lower bounds for size of modal definitions over $\TT^2$ cannot work for DAG-size. The same lower bound construction fails for DAG-size over unrestricted models, although there the exact DAG-size complexity of optimal separators remains unknown.


%

\subsection{Open Problems.} We mention some interesting open problems. First, the relative
succinctness of $\muML$ over $\ML$ is to the best of our knowledge
open in the setting with only one accessibility relation. Second, as we have
mentioned in Section~\ref{sec:foundations}, the separators we compute
are not necessarily the logically strongest ones. The logically
strongest separators of $\varphi,\phi'$ are precisely the
$\ML$-uniform consequences of $\phi$
(if they exist) and are a natural object of study. 
 Clearly, modal
definability of $\varphi$ is a sufficient condition, but 
not a necessary one. Let $\varphi=\psi\wedge\neg \theta_\infty$ and
$\varphi'=\neg\psi$ for some $\psi\in \ML$ and $\theta_\infty$ from Example~\ref{exa:separability}. Then $\varphi$ is not
equivalent to a modal formula, but $\psi$ is a strongest separator. It would also be interesting to investigate separability of $\muML$-formulae by other fragments of $\muML$, such as \emph{safety formulae}, for which definability has been studied in~\cite{DBLP:conf/csl/LehtinenQ15} and to extend the $\mu$-calculus by other operators, e.g., nominals or inverse modalities. 

Finally, it would be interesting to extend our findings in the case study on the graded $\muML$ in Section~\ref{sec:graded}. As already pointed out, there are two natural directions for this: first one could study the complexity under binary coding of grades, and second one could tackle the separator construction problem. We conjecture that the first problem can be solved by combining our techniques with the automata constructions from~\cite{DBLP:conf/cade/KupfermanSV02}, which lead to optimal complexity for satisfiability of $\grML$ under binary coding of grades. The computation problem might be much more challenging. Indeed, it was observed recently that the construction of $\ML$ separators of $\grML$-formulae already requires significantly new ideas~\cite{DBLP:journals/corr/abs-2507-15689}.

\bibliography{modsep}

\end{document}